\documentclass[oneside,article]{memoir}

\usepackage{amsmath}         % align-like environments
\usepackage{amsthm,thmtools,thm-restate} % theorem-like environments
\usepackage{enumitem}        % enumerate-like environments

\usepackage{rotating}

\usepackage[final]{listings}
\usepackage{bbm}
\usepackage{xcolor}

\definecolor{codegreen}{rgb}{0,0.6,0}
\definecolor{codegray}{rgb}{0.5,0.5,0.5}
\definecolor{codepurple}{rgb}{0.58,0,0.82}
\definecolor{backcolour}{rgb}{0.95,0.95,0.92}

\lstdefinestyle{mystyle}{
    backgroundcolor=\color{backcolour},   
    commentstyle=\color{codegreen},
    keywordstyle=\color{magenta},
    numberstyle=\tiny\color{codegray},
    stringstyle=\color{codepurple},
    basicstyle=\ttfamily\footnotesize,
    breakatwhitespace=false,         
    breaklines=true,                 
    captionpos=b,                    
    keepspaces=true,                 
    numbers=left,                    
    numbersep=5pt,                  
    showspaces=false,                
    showstringspaces=false,
    showtabs=false,                  
    tabsize=2
}

\usepackage{algorithm}
\usepackage{algpseudocode}
\usepackage{algorithmicx}
\usepackage{threeparttable}
\usepackage{float}

\usepackage{tikz}
\usetikzlibrary{matrix,arrows,positioning}
\usepackage{amsfonts}
\usepackage{parskip}
\usepackage{mathtools}
\usepackage{wrapfig}

\usepackage{tikz-cd}

\usepackage{caption, subcaption} % for subfigures

\usepackage[hidelinks]{hyperref}                   % hyperlinks
\usepackage[capitalize,nosort]{cleveref}           % named references

\usepackage{amssymb}
\usepackage[mathscr]{euscript}
\usepackage{relsize}
\usepackage{stmaryrd}
\usepackage{stackengine}

\usepackage{indentfirst} % indent the first paragraph
\usepackage{multicol}

\usepackage{graphicx}
\usepackage{ragged2e} % for nicer left alignment inside columns
\usepackage[inline,nomargin,draft]{fixme}

\isopage[12]

\setlrmargins{*}{*}{1}
\checkandfixthelayout

\counterwithout{section}{chapter}
\usepackage{appendix}

\makeatletter
\providecommand{\abx@aux@refcontext}[1]{}
\providecommand{\abx@aux@cite}[2]{}
\providecommand{\abx@aux@segm}[3]{}
\providecommand{\abx@aux@page}[2]{}
\providecommand{\abx@aux@fnpage}[2]{}
\providecommand{\abx@aux@backref}[5]{}
\providecommand{\abx@aux@defaultrefcontext}[3]{}
\providecommand{\abx@aux@read@bbl@mdfivesum}[1]{}
\providecommand{\abx@aux@read@bblrerun}{}
\providecommand{\abx@aux@sortscheme}[1]{}
\providecommand{\abx@aux@refsection}[1]{}
\providecommand{\abx@aux@number}[2]{}
\makeatother

\makeatletter
\let\dgm\@undefined
\makeatother

  \newcommand{\id}[1][]{\operatorname{id}_{#1}}
  
  \declaretheorem[style=definition,within=section]{definition}
  \declaretheorem[style=definition,numberlike=definition]{example}
  \declaretheorem[style=definition,numberlike=definition]{remark}
  \declaretheorem[style=definition,numberlike=definition]{notation}

  \declaretheorem[style=plain,numberlike=definition]{lemma}
  \declaretheorem[style=plain,numberlike=definition]{proposition}
  \declaretheorem[style=plain,numberlike=definition]{theorem}

  \declaretheorem[style=plain,numbered=no,name=Theorem]{theorem*}

  \Crefname{corollary}{Corollary}{Corollaries}
  \Crefname{definition}{Definition}{Definitions}
  \Crefname{lemma}{Lemma}{Lemmas}
  \Crefname{proposition}{Proposition}{Propositions}
  \Crefname{remark}{Remark}{Remarks}
  \Crefname{theorem}{Theorem}{Theorems}
  \Crefname{notation}{Notation}{Notations}
  \Crefname{conjecture}{Conjecture}{Conjectures}

  \newlist{axioms}{enumerate}{1}
  \Crefname{axiomsi}{}{}

  \newenvironment{tikzeq*}
  {
    \begingroup
    \begin{equation*}
    \begin{tikzpicture}[baseline=(current bounding box.center)]
  }
  {
    \end{tikzpicture}
    \end{equation*}
    \endgroup
    \ignorespacesafterend
  }

  \tikzset
  {
    diagram/.style=
    {
      matrix of math nodes,
      column sep={4.3em,between origins},
      row sep={4em,between origins},
      text height=1.5ex,
      text depth=.25ex
    },
    over/.style={preaction={draw=white,-,line width=6pt}},
    every to/.style={font=\footnotesize},
    inj/.style={right hook->},
    surj/.style={-{Latex[open]}},
    cof/.style={>->},
    fib/.style={->>},
  }

  \DeclareFontFamily{U}{mathx}{\hyphenchar\font45}

  \DeclareFontShape{U}{mathx}{m}{n}{
    <5> <6> <7> <8> <9> <10>
    <10.95> <12> <14.4> <17.28> <20.74> <24.88>
    mathx10}{}

  \DeclareSymbolFont{mathx}{U}{mathx}{m}{n}

  \DeclareFontFamily{U}{mathb}{\hyphenchar\font45}

  \DeclareFontShape{U}{mathb}{m}{n}{
    <5> <6> <7> <8> <9> <10>
    <10.95> <12> <14.4> <17.28> <20.74> <24.88>
    mathb10}{}

  \DeclareSymbolFont{mathb}{U}{mathb}{m}{n}

  \DeclareMathAccent{\widebar}{0}{mathx}{"73}

  \DeclareMathSymbol{\Rsh}{\mathrel}{mathb}{"E9}

  \DeclareFontFamily{U}{MnSymbolA}{}

  \DeclareFontShape{U}{MnSymbolA}{m}{n}{
    <-6> MnSymbolA5
    <6-7> MnSymbolA6
    <7-8> MnSymbolA7
    <8-9> MnSymbolA8
    <9-10> MnSymbolA9
    <10-12> MnSymbolA10
    <12-> MnSymbolA12}{}

  \DeclareSymbolFont{MnSyA}{U}{MnSymbolA}{m}{n}

  \DeclareMathSymbol{\twoheaddownarrow}{\mathrel}{MnSyA}{27}

  \newcommand{\MSC}[1]{%
    \let\thempfn\relax
    \footnotetext[0]{2020 Mathematics Subject Classification: #1.}
  }

  \newcommand{\VR}{\textsf{VR}}
  \newcommand{\dgm}{\textsf{dgm}}

  \newcommand{\keys}{\operatorname{keys}}
  \newcommand{\filt}{\operatorname{filt}}

  \newcommand{\bigtriangle}{\scalebox{1.5}{$\triangle$}}

\tikzstyle{vertex}=[circle, draw, minimum size=7pt, inner sep=0pt]

\providecommand{\VR}{}
\renewcommand{\VR}{\mathsf{VR}}
\newcommand{\redzed}{{\textsc{RedZeD}}}
\newcommand{\redzedvr}{{\textsc{RedZeD{\textunderscore}VR}}}
\newcommand{\bbF}[1]{\mathbb{F}_{#1}}

\newcommand{\Ch}{\mathsf{Ch}}
\newcommand{\Vect}{\mathsf{Vect}}

\DeclareFontFamily{U}{dmjhira}{}
\DeclareFontShape{U}{dmjhira}{m}{n}{ <-> dmjhira }{}

\author{Krzysztof Kapulkin \and Nathan Kershaw}

\title{RedZeD: Computing persistent homology by Reduction to Zero Differentials}

\date{\today}

\begin{document}

\maketitle

\begin{abstract}
  We introduce a new algorithm for computing persistent homology of Vietoris--Rips filtrations, which offers a considerable improvement both in terms of time and memory over the existing implementations of the persistence pairing algorithm.
  The key innovation, called active enumeration, is made possible by a new theoretical framework of Reduction to Zero Differentials (hence RedZeD) in which to view persistent homology.
  
\end{abstract}

%  \setlist[enumerate]{label=(\arabic*)}

% Add content here

\section*{Introduction}

Persistent homology is a fundamental tool of topological data analysis, and, more broadly, of computational topology \cite{zomorodian-carlsson:computing-persistent-homology,edelsbrunner-letscher-zomorodian:persistence,otter-et-al:roadmap}.
It takes as input a filtration (or often a more general diagram) of chain complexes and returns a persistence barcode which tracks topological features of that input filtration.
Examples of filtrations of chain complexes abound in both pure and applied mathematics, arising from (filtrations of) topological spaces, simplicial and cubical complexes, simplicial and cubical sets, graphs, and data.
Consequently, considerable effort has been put into designing algorithms for efficiently computing this invariant \cite{chen-kerber:twist,silva-morozov-vejdemo-johansson:dualities,clear-and-compress,Boissonnat-dey-maria:compressed-annotation,phat-original,otter-et-al:roadmap,bauer:ripser,keeping-it-sparse}.

The Vietoris--Rips filtrations arising from finite metric spaces, or more generally sets equipped with a distance function (i.e., a function $d \colon X \times X \to \mathbb{R}_+$ that is symmetric and satisfies $d(x, x) = 0$), are of particular importance in topological data analysis and topological inference.
Given a finite metric space, its Vietoris--Rips filtration is a filtration of simplicial complexes with vertex set $X$, and with higher simplices at scale $r \geq 0$ given by finite subsets of points that are pairwise at most $r$ apart.
The persistent homology of this filtration can be used to detect topological features of the input that persist through different time scales.

The \emph{persistence pairing} algorithm \cite{zomorodian-carlsson:computing-persistent-homology}, also referred to as the \emph{standard} algorithm, is the basic tool for computing persistent homology of a filtration of chain complexes.
It uses restricted Gaussian elimination, reducing columns left-to-right by adding previously reduced columns to either clear a column or identify a unique pivot.
Since its introduction, many improvements to the algorithm have been proposed.
Broadly speaking, these improvements fall into two categories: reducing the size of chain complexes involved in the filtration at hand, or providing more efficient tools for generating and reducing the matrix.
Unsurprisingly, amongst the broad array of these tools, there is significant variance when it comes to the generality they apply in and how much speedup they offer.

In the specific case of the Vietoris--Rips filtration mentioned above, the software Ripser \cite{bauer:ripser} has emerged as the gold standard for efficient computation of persistent homology.
By combining different improvements such as: clearing birth columns \cite{chen-kerber:twist}, computing cohomology instead of homology \cite{silva-morozov-vejdemo-johansson:dualities}, implicit matrix reductions, and identifying apparent and emerging pairs, Ripser outperforms other im\-ple\-men\-ta\-tions both in terms of computational speed and memory usage.

The purpose of the present work is twofold.
First, we provide an abstract rephrasing and a generalization to the level of chain complexes of some of the existing algorithms.
Second, we introduce a new speedup that is specific to persistent homology of Vietoris--Rips complexes.

The first of these goals is accomplished by placing the persistence pairing algorithm in a more abstract context.
Specifically, we show that it can be seen as replacing a given filtration of chain complexes (over a field) by a naturally quasi-isomorphic sequence of chain complexes with zero differentials.
Let us unpack this statement.
If a chain complex has zero differentials, its homology groups are simply the chain groups themselves.
Thus, given a filtration of chain complexes $C^*$, we might wish to replace it with a new sequence of chain complexes $A^*$ (which might no longer be a filtration of subcomplexes) with zero differentials, along with a family of quasi-isomorphisms $r^* \colon C^* \to A^*$, i.e., maps inducing isomorphisms on homology.
The key insight is that each chain complex $A^i$ in the sequence $A^*$ depends only on the previous chain complex $A^{i-1}$ and the map $r^{i-1} \colon C^{i-1} \to A^{i-1}$.
At its core, this observation is not new and was made in the context of cochain complexes arising from simplicial complexes in \cite{Dey-Fan-Wang}, where the map $r$ mentioned above is called an \emph{annotation}.
The novel contribution of our work up to this point is a more abstract treatment of this insight.

Algorithmically, at each step of the filtration, we can store $r$ using dictionaries $R_0$, $R_1$, \ldots, $R_n$ corresponding to the homological degree.
With this, a straightforward implementation of the proof of the aforementioned fact leads to the first algorithm considered in the present paper, named \redzed, which is an acronym for \textbf{Red}uction to \textbf{Ze}ro \textbf{D}ifferentials.
While seemingly approaching the problem from a different direction, \redzed{} can be rephrased in the language of matrix reduction algorithms, in which case it can be seen to implement several previously studied improvements to the standard algorithm, namely, \emph{compression} \cite{clear-and-compress} along with \emph{exhaustive and retrospective reduction} \cite{keeping-it-sparse, edelsbrunner-olsbock,edelsbrunner-zomorodian,zomorodian-carlsson:computing-persistent-homology}.
In other words, \redzed{} represents both a generalization and a more abstract treatment, along with a new implementation of several existing ideas.

On its own, \redzed{} does not offer significant improvements over existing methods.
This new perspective, however, suggests a new improvement to the algorithm specific to the Vietoris--Rips filtration. 
This improvement is only possible on top of \redzed{}, making it a necessary component.
Typically, when computing persistent homology of Vietoris--Rips complexes up to degree $n$, the majority of time is spent reducing $(n+1)$-birth columns, i.e., columns creating new features in degree $n+1$, but not affecting degree $n$.
Consequently, one wishes to avoid considering such $(n+1)$-simplices altogether, or at the very least ``weed out'' as many of them as possible.
In Ripser, this is done by computing cohomology and birth clearing, neither of which apply to our homological setting.
Instead, we employ a technique which we refer to as \emph{active enumeration}.
We say that a birth simplex $\sigma$ is \emph{active} at time $i$ if $r^i(\sigma) \neq 0$, and it is \emph{inactive} otherwise. 
This cuts across the standard distinction between alive and dead simplices used in persistence, as now a simplex can be simultaneously both dead and active.

After adding a new birth $n$-simplex $\sigma$, active enumeration searches for all $(n+1)$-simplices $\tau$ that have at least one active face and $\sigma$ as their youngest face.
This can be done by an efficient check: loop through all active $n$-simplices, check whether they share an $(n-1)$-face with $\sigma$, then check whether a single vertex can be added to $\sigma$ to form a desired $(n+1)$-simplex.
The key point is that \redzed's use of dictionaries $R_0$, $R_1$, \ldots, $R_n$ makes it immediate to identify active $n$-simplices --- they are exactly the keys of $R_n$.
This constitutes \redzedvr, a specialization of \redzed{} to the Vietoris--Rips filtration.

To ascertain the quality of improvements implemented in \redzedvr, relative to other software, notably Ripser, we run several tests comparing their performance across different data sets.
These include both curated data sets with known characteristics, and common benchmark data sets found in \cite{bauer:ripser}.
We test 22 different cases, and find that \redzedvr{} is faster in 21 of them, the one outlier being $H_3$ of a noisy 3-dimensional sphere. 
In many cases, \redzedvr{} also scales better than Ripser. 
These results are most pronounced for $H_1$, especially on data with more topological signal or high-dimensional data.
Speedups in higher homology degrees are not as pronounced as in degree 1, but \redzedvr{} generally remains faster, aside from the one aforementioned case. 
We also find that \redzedvr{} generally uses less peak memory than Ripser. 
Again, these results are most pronounced in degree 1, where \redzedvr{} always uses less memory. 
\redzedvr{} generally uses less memory in higher degrees as well; however, Ripser tends to outperform \redzedvr{} in higher degrees for data with little topological signal, or for large datasets computed with a low threshold.
Some of the more notable results for $H_1$ can be seen in the table below:

\begin{table}[htbp] \label{table:timemem}
  \centering
  \setlength{\tabcolsep}{4pt}
  \begin{tabular}{llrrrrrrr}
    \toprule
    & & & \multicolumn{3}{c}{Time (s)} & \multicolumn{3}{c}{Peak memory (MiB)} \\
    \cmidrule(lr){4-6} \cmidrule(lr){7-9}
    Dataset & & pts & \redzedvr{} & Ripser & Speedup & \redzedvr{} & Ripser & Ratio \\
    \midrule
    grid spheres, $S^{1}$ & $H_1$ & 5,000 & 0.706 & 171 & 242 & 167 & 452 & 2.71 \\
    noisy sphere, $S^{1}$ & $H_1$ & 3,800 & 3.87 & 292 & 75.4 & 231 & 16,500 & 71.4 \\
    random distance & $H_1$ & 450 & 0.404 & 251 & 621 & 30.9 & 2,070 & 66.9 \\
    random euclidean, $\mathbb{R}^{10}$ & $H_1$ & 6,000 & 1.12 & 86.5 & 76.9 & 258 & 4,400 & 17.1 \\
    \texttt{o3\_1024} & $H_1$ & 1,024 & 0.0599 & 0.898 & 15.0 & 16.4 & 145 & 8.82 \\
    \texttt{dragon} & $H_1$ & 2,000 & 0.114 & 0.969 & 8.51 & 30.0 & 98.4 & 3.28 \\
    \texttt{o3\_4096} & $H_1$ & 4,096 & 1.40 & 53.3 & 38.0 & 177 & 8,360 & 47.2 \\
    \bottomrule
  \end{tabular}
  \caption{Time and memory comparisons between \redzedvr{} and Ripser for various datasets}
\end{table}

\subsection{Organization of the paper}

The paper is organized into six sections.
\cref{sec:prelims} treats preliminaries from persistent homological algebra and topological data analysis.
Readers familiar with these topics can safely skip this section.

In \cref{sec:algo}, we introduce the \redzed{} algorithm by first providing a theoretical result on replacing a filtration by a quasi-isomorphic sequence, and then extracting an algorithm, called ``Naive \redzed{}'', from it.
Naive \redzed{} is a naive implementation of reducing a filtered chain complex to a sequence with zero differentials, which we use for ease of proving correctness and extracting the interpretation in terms of the standard matrix reduction algorithm.
We conclude \cref{sec:algo} by discussing simple improvements that speed up our algorithm, including explicitly storing inverses of the dictionaries $R_p$ and deleting keys when the corresponding values are zero.
We reserve the name \redzed{} to refer to the algorithm with these improvements implemented, distinguishing it from the Naive \redzed{} algorithm.

In \cref{sec:active}, we describe active enumeration and how it can be used to speed up computations of Vietoris--Rips persistent homology.
In particular, we show how active simplices can be understood in the matrix-theoretic language.
This gives \redzedvr, a specialization of \redzed{} to the Vietoris--Rips filtration.
The \redzedvr{} algorithm is implemented in both Julia and C++, and details on the implementation are described in \cref{sec:implementation}.
We then compare \redzedvr{} to Ripser in \cref{sec:experiments}, reporting on relative differences in performance.
We conclude in \cref{sec:conclusion} with a summary of the paper and a discussion of future directions.

The code for \redzedvr{} is available on
\begin{center}
    \url{https://github.com/nkershaw01/RedZeD}
\end{center}
and can also be downloaded as a Python package called \texttt{redzed-tda}:
\begin{center}
    \url{https://pypi.org/project/redzed-tda/}.
\end{center}

\subsection{Related work}
As mentioned before, the results and algorithms of \cref{sec:algo} can be viewed as a new (and more general) interpretation of part of the algorithm presented in \cite{Dey-Fan-Wang}, which concerns persistent homology of arbitrary simplicial maps (not necessarily inclusions).
\cite{Dey-Fan-Wang} works with cohomology and keeps track of the cohomological basis using the \emph{annotation} map, which is analogous to the map $r$ described above.
The behavior of this algorithm on \emph{elementary inclusions} corresponds directly to the steps performed in \redzed.
We comment on this relationship in detail in \cref{rmk:dey-fan-wang}.

In the language of matrix reduction algorithms, the techniques employed by \redzed{} include compression along with exhaustive and retrospective reduction.
These have been studied extensively and we refer the reader to \cite{keeping-it-sparse} for an excellent overview and comparison.

In terms of performance, we compared \redzedvr{} to Ripser \cite{bauer:ripser}, which represents the high water mark in computing persistent homology of Vietoris--Rips filtrations.
In the implementation of \redzedvr, we also borrow several ideas from Ripser, notably its combinatorial numbering system used to encode simplices.

\subsection{Acknowledgements}

As newcomers to the area, we have greatly benefited from conversations with H{\aa}vard Bakke Bjerkevik, Daniel Carranza, Tamal Dey, Barbara Giunti, Mike Lesnick, Dmitriy Morozov, and Luis Scoccola.
We are grateful to all of them for patiently answering our questions and providing excellent suggestions and insights.

\subsection{AI Disclaimer}

After completing the first implementation, we used Claude Opus 5 to optimize the code; in particular, to create a new implementation of dictionaries, with backward shift deletion instead of using tombstones.
We also used Claude to translate our code from Julia to C++.

\section{Preliminaries} \label{sec:prelims}

In this section, we review the necessary background on (persistent) homological algebra.
Throughout, we work over a field $F$, suppressing it from notation, thus, for example, writing $\Ch$ for the category of chain complexes and chain maps.
Similarly, we write $H_n(X)$ to mean the $n$-th homology group of a simplicial complex $X$ with coefficients in $F$, more commonly denoted $H_n(X; F)$.

In the context of persistence, our main object of study will be filtrations of chain complexes.
Although we index our filtrations by natural numbers, we could alternatively choose any finite totally ordered set.
 
\begin{definition} \label{def:filtration}
    A (finite) \emph{filtration} $C^*$ of chain complexes is a sequence of inclusions of chain complexes
    \[C^0 \subseteq C^1 \subseteq \dots \subseteq C^n \text{.}\]
\end{definition}

To make our filtrations well-behaved, we will impose the condition that our filtrations are \emph{elementwise}, i.e., the ``difference'' between $C^i$ and $C^{i+1}$ is a single generator.
To express it in the categorical language, which will be useful to us a little later, we define the algebraic models of the $n$-sphere and the $n$-disk for a non-negative integer $n$.
The \emph{algebraic $n$-sphere} $S(n)$ and the \emph{algebraic $n$-disk} $D(n)$ is the chain complexes given by
    \[
    S(n)_p = 
    \begin{cases}
    F & \text{if } n=p \\
    0 & \text{otherwise}
    \end{cases} 
    \qquad \text{ and } \qquad
    D(n)_p = 
    \begin{cases}
    F & \text{if } n=p-1, p \\
    0 & \text{otherwise.}
    \end{cases}\]
For every $n$, there is a natural inclusion $S(n-1) \hookrightarrow D(n)$.
With that as a preamble, we can now define elementwise filtrations.

\begin{definition} \label{def:elementwise}
    A filtration $C^*$ is \emph{elementwise} if $C^0 = 0$, and for all $1 \leq i \leq n$ $C^i$ can be constructed via a pushout square:
    \begin{center}
    \begin{tikzcd}
    S(n-1) \arrow[d, swap, hook] \arrow[r] & C^{i-1} \arrow[d, hook] \\
    D(n) \arrow[r]                         & C^i               
    \end{tikzcd}
    \end{center}
\end{definition}
Explicitly, in an elementwise filtration, we have for every $i \geq 1$ there is a $p$ such that:
    \[
    (C^i / C^{i-1})_n = 
    \begin{cases}
    F & \text{if } n=p \\
    0 & \text{otherwise.}
    \end{cases}\] 
If we have a choice of basis preserved by the inclusions, we say that the generator $x$ of $F$ in dimension $p$ is the \emph{element added at time} $i$, and $p$ is called the \emph{dimension} of $x$. 
Let $C^*$ be an elementwise filtration, and $x$ a basis element of $C^i_p$ for some $i$ and $p$. 
The \emph{filtration time} of $x$, denoted $\filt(x)$, is the minimal index $j$ such that $x \in C^j_p$. 

As chain complexes are typically studied using homology, we use persistence modules to study filtrations of chain complexes.
Since our filtrations are indexed by finite totally ordered sets, we will likewise consider persistence modules over finite totally ordered sets, obtained by taking the homology of chain complexes in the filtration.

\begin{definition} \label{def:pers_module}
    A \emph{persistence module} over a totally ordered set $T$ is a functor $V \colon T \to \Vect_F$ from $T$, viewed as a category, to the category of $F$-vector spaces. 
\end{definition}

Explicitly, a persistence module $V$ over $T$ consists of a collection of $F$-vector spaces $\{V_i\}_{i \in T}$ together with maps $\varphi_i^j\colon V_i \to V_j$ for all $i \leq j$, such that each $\varphi_i^i = \id$ and for all $i \leq j \leq k$, we have $\varphi_i^k = \varphi^k_j \circ \varphi^j_i$.

An important example of a persistence module is the notion of an interval module:
\begin{example} \label{ex:interval_module}
    Let $T$ be a totally ordered set, and $b \leq d \in T$. 
    The \emph{interval module}, denoted $I[b,d)$, is the persistence module where $I[b,d)_i$ is $F$ for $b \leq i < d$, and $0$ elsewhere.
    For $i \leq j$, $\varphi_i^j = \id$ for $i \in [b,  d)$, and the zero map otherwise.
\end{example}

Given a persistence module over a finite set $V$, the following classical theorem of P.~Gabriel allows us to decompose it into interval modules in a canonical way:

\begin{theorem}[Gabriel, \cite{gabriel:decomposition}] \label{TH:decomposition}
    For a finite totally ordered set $T$ and a persistence module $V$ over $T$, there exists a unique decomposition of $V$ into a direct sum of interval modules:
    \[ V \cong \bigoplus_J I[b_j,d_j)\text{.} \tag*{\qed} \]
\end{theorem}

\begin{definition}
    Let $C^*$ be a filtration of chain complexes, and $n$ a natural number. 
   The \emph{persistence diagram} of $C^*$ in dimension $n$ is the multiset of pairs $ \dgm_n(C^*)=\{[b_j,d_j)\}_{j \in J}$ associated to the unique decomposition $H_n(C^*) \cong \bigoplus_J I[b_j,d_j)$.
   A \emph{persistence pair} is an element of $\dgm_n(C^*)$. 
\end{definition}

Note that $\dgm_n(C^*)$ is a multiset, rather than just a set, because an interval $I[b,d)$ may appear multiple times in the decomposition. 
The computation of the persistence diagram of a filtration of chain complexes is often referred to as \emph{persistent homology}. 
The standard method to compute the persistence diagram is the \emph{persistence pairing algorithm}, \cite{zomorodian-carlsson:computing-persistent-homology,silva-morozov-vejdemo-johansson:dualities,bauer:ripser}, which is a matrix reduction algorithm. 

The input to the persistence pairing algorithm is the \textit{filtration boundary matrix} $M$. 
The matrix $M$ is an $n \times n$ matrix, with rows and columns representing elements in the filtration, ordered by the order in which they appear.
Let $x$ and $y$ denote elements in the filtration. 
We write $M_y$ to mean the column corresponding to the element $y$, and $M_{xy}$ to be the entry in the row corresponding to $x$ and the column corresponding to $y$.
If $x \in \partial y$, the entry $M_{xy}$ corresponds to the coefficient of $x$ in $\partial y$, and otherwise $M_{xy}=0$.
This way, the column $M_y$ represents $\partial y$.

The matrix $M$ is reduced column by column, from left to right. 
A column $M_y$ is reduced only by adding scalar multiples of columns to the left of it, and it is reduced until it is either zero or has a pivot unique from all previous columns. 
Once the entire matrix is reduced, the persistence pairs can be read off from the reduced matrix. 
If an entry $M_{xy}$ is a pivot of column $y$, this corresponds to a persistence pair $(x,y)$. 
A zero column $M_y$ such that $y$ does not appear as the pivot row of any other column corresponds to a pair $(y, \infty)$. 

Let $C^*$ be an elementwise filtration and suppose that $C^i$ was obtained from $C^{i-1}$ by adding a generator $x$ of dimension $p$.
This implies that either: the dimension of the image of $\partial_{p}$ increased by $1$, or the dimension of the kernel of $\partial_p$ increased by $1$. 
In the former case, $\dim H_{p-1} (C^i) = \dim H_{p-1} (C^{i-1}) - 1$ and we call $x$ a \emph{death element}. 
In the latter case, $\dim H_{p-1} (C^i) = \dim H_{p-1} (C^{i-1}) + 1$ and we call $x$ a \emph{birth element}. 
Equivalently, $x$ is a birth element if the column corresponding to $x$ in the reduced matrix is zero, and a death element otherwise.

Persistent homology is most commonly used to compute Vietoris--Rips homology of finite metric spaces.
The Rips complex construction associates a filtration of simplicial complexes to a finite metric space.
However, the construction is more general than that; namely, it works for an arbitrary set $X$ equipped with a symmetric distance function $d$, i.e., a function  $d \colon X \times X \to [0, \infty)$ such that for all $x, y \in X$, we have $d(x, x) = 0$ and $d(x, y) = d(y, x)$.
Compared to the definition of a metric, we are therefore dropping the triangle inequality.
Such generality is useful in a variety of applications, including those considered in \cite{kapulkin-kershaw:data-analysis}.

\begin{definition}
    Let $(X,d)$ be a finite set equipped with a distance function, and $r \in [0, \infty)$.
    The \emph{Vietoris--Rips complex} at scale $r$ is the simplicial complex $\VR_r(X)$ with:
    \begin{itemize}
        \item $\VR_r(X)_0 = X$
        \item $\VR_r(X)_n = \{\{x_0,\dots, x_n\} \ | \ d(x_i,x_j) \leq r \text{ for all } 0 \leq i,j \leq n \} $
    \end{itemize}
Letting $r$ vary, we can form a filtration of simplicial complexes called the \emph{Vietoris--Rips filtration}, denoted $\VR(X)$. 
\end{definition}
The complex $\VR_r(X)$ can also be viewed as the clique complex on the graph $G_r(X)$, which has as vertices $X$ and an edge $x \sim y$ if $d(x,y) \leq r$. 
With this viewpoint, we can also generalize the Vietoris--Rips filtration to any filtration of graphs. 

Let $C$ be the functor sending a simplicial complex $K$ to the chain complex $C(K)$, which in dimension $n$ is the $F$-vector space generated by the set of $n$-simplices. 
Applying this to each element of the filtration $\VR(X)$, we get a filtration of chain complexes $C(\VR(X))$.

\begin{definition}
    The \emph{Vietoris--Rips persistence diagram} of a set  $(X, d)$ equipped with a symmetric distance function is the persistence diagram $\dgm_{\VR}(X) = \dgm(C(\VR(X))$. 
\end{definition}
\newcommand{\refined}{\mathsf{ref}}
The Vietoris--Rips filtration is not necessarily an elementwise one, but it can be refined to one.
To do this, we impose (arbitrarily) a total order on the elements of $X$. 
Then, we can order the simplices on the complete simplicial complex on $X$ lexicographically according to
\begin{enumerate}
    \item the diameter (i.e., maximum length of an edge in the simplex) with $\sigma < \tau$ if $\operatorname{diam}(\sigma) < \operatorname{diam}(\tau)$;
    \item the dimension with $\sigma < \tau$ if $\operatorname{diam}(\sigma) = \operatorname{diam}(\tau)$, and $\dim(\sigma)<\dim(\tau)$;
    \item ordering on the vertices with $\sigma \leq \tau$ if their diameters and dimensions are equal, and the list of vertices of $\sigma$ are lexicographically less than or equal to the list of vertices of $\tau$.
\end{enumerate}
We can then form a filtration indexed by the natural numbers, where at time $i$ we add the $i$-th simplex according to the order. 
We call this filtration the \emph{lexicographically refined Vietoris--Rips filtration} and denote it $\VR_\refined(X)$, thus making the filtration $C(\VR_\refined(X))$ elementwise as well.
In practice, one thus computes $\dgm(C(\VR_{\refined}(X))$ and uses it to reconstruct $\dgm_{\VR}(X)$. 
This is done by replacing each $(i,j)$ with $(\operatorname{diam}(\sigma_i),\operatorname{diam}(\sigma_j))$, where $\sigma_i$ and $\sigma_j$ are the simplices added at times $i$ and $j$, and then removing all intervals of length 0. 
Because of this, we will always assume that all the Vietoris--Rips filtrations are already refined, and will omit the subscript $(-)_\refined$ going forward.

Similar to the terminology for chain complexes, in the refined filtration, we call a simplex $\sigma$ a \emph{birth simplex} if its corresponding element in the filtration is a birth element.
We call a $\sigma$ a \emph{death simplex} if its corresponding element is a death element. 
\section{The \redzed{} algorithm} \label{sec:algo}

In this section, we introduce our general purpose algorithm \redzed{} for computing persistent homology of an arbitrary elementwise filtration of chain complexes over $\bbF 2$.
The name \redzed{} is a syllabic abbreviation and stands for \textbf{Red}uction to \textbf{Ze}ro \textbf{D}ifferentials.
Although stated for $\bbF 2$, the algorithm can be easily modified to work over an arbitrary field, which we comment on in \cref{rmk:works_over_field}.
Note that while we are referring to the algorithm presented in this section as \redzed, in \cref{sec:active} we introduce a further improvement, called \emph{active enumeration}, which specializes \redzed{} to Vietoris--Rips filtration and yields a new algorithm \redzedvr.
On its own, the \redzed{} algorithm does not offer an improvement over existing methods, and is in general significantly slower. 
Without \redzed{}, however, the active enumeration speedup is not possible, making \redzed{} a necessary component. 

We do not claim any originality in developing the \redzed{} algorithm; as indicated above, it is merely a rephrasing of part of the algorithm in \cite{Dey-Fan-Wang}, which we comment on in more detail in \cref{rmk:dey-fan-wang}. 
Similarly, while \redzed{} might not at first appear to be performing matrix reductions of the persistence pairing algorithm, it can be translated into the matrix-theoretic language.
Upon such translation, discussed later in the section, it recovers several of the well-known improvements to the standard algorithm, namely, compression \cite{clear-and-compress} along with exhaustive \cite{edelsbrunner-zomorodian,zomorodian-carlsson:computing-persistent-homology,edelsbrunner-olsbock} and retrospective reduction \cite{keeping-it-sparse}.

In the first part of this section, we introduce Naive \redzed, the version of the algorithm that directly follows from \cref{prop:reduction}, and prove its correctness. 
In the second part of the section, we translate Naive \redzed{} to the language of matrices and discuss the comparisons with existing improvements to the standard algorithm. 
In the final part of this section, we discuss two minor improvements to Naive \redzed. When these improvements are implemented, we call the algorithm \redzed{}.

% While we can translate the algorithm into another form of the persistence pairing algorithm, which we comment on in \cref{rmk:redzed_is_matrix}, we first present it in a way independent of the persistence pairing algorithm and don't explicitly work with matrices.

Recall that a chain map $f \colon C \to D$ between chain complexes is a \emph{quasi-isomorphism} if it induces isomorphisms $H_n(C) \to H_n(D)$ for every non-negative integer $n$.
The key theorem underlying our algorithm says that an elementwise filtration of chain complexes can be replaced by quasi-isomorphic one in which all chain complexes have zero differentials.
More precisely, it is the procedure given in the proof of this statement that underpins our algorithm.

\begin{theorem} \label{prop:reduction}
    Let $F$ be a field and $C^* \colon 0 = C^0 \hookrightarrow C^1 \hookrightarrow \dots \hookrightarrow C^m$ be an elementwise filtration of chain complexes over $F$. 
    Then there exists a diagram 
    \begin{center}
    \begin{tikzcd}
    C^0 \arrow[r] \arrow[d, "r^0"' , "\simeq"] 
      & C^1 \arrow[r] \arrow[d, "r^1"', "\simeq"] 
      & \cdots \arrow[r] 
      & C^m \arrow[d, "r^n"', "\simeq"] \\
         A^0 \arrow[r] 
      & A^1 \arrow[r] 
      & \cdots \arrow[r] 
      & A^n
    \end{tikzcd}
    \end{center}
    where each $r^i$ is a quasi-isomorphism, and the maps $\partial^i_j:A^i_j \to A^i_{j-1}$ are the zero map for all $i$ and $j$. 
\end{theorem}

The proof of the theorem depends on a stability property of quasi-isomorphisms in the category of chain complexes.
Recall that the category of chain complexes carries the (injective) model structure whose cofibrations are the monomorphisms and weak equivalences are the quasi-isomorphisms \cite[Thm.~2.3.13]{hovey}.
Since all objects are injectively cofibrant, this model structure is left proper \cite[Prop.~13.1.2]{hirschhorn}, i.e., pushouts of quasi-isomorphisms along monomorphisms are again quasi-isomorphisms.

\begin{lemma} \label{prop:left_properness}
    Let $f\colon C \to D$ be a quasi-isomorphism and $C \hookrightarrow C'$ a levelwise injective morphism.
    Let $f'$ be the pushout of $f$ along this inclusion:
    \begin{center}
    \begin{tikzcd}
    C \arrow[d, "f", swap] \arrow[r, hook] & C' \arrow[d, "f'"] \\
    D \arrow[r]                      & D'                
    \end{tikzcd}
    \end{center}
    Then $f'$ is again a quasi-isomorphism. \qed
\end{lemma}

With the lemma in hand, we can now give the proof of \cref{prop:reduction}.
 
\begin{proof}[Proof of \cref{prop:reduction}]
    Since each $C^i$ is quasi-isomorphic to $A^i$, we must have that $A^i \cong H_*(C^i)$, viewed as a chain complex with zero differential. 
    The important step is constructing the maps $r^i$. 

    We proceed by induction. 
    We have that $C_0 = A_0=0$, and $r^0 = \id[0]$. 
    Now suppose we have $r^{i-1} \colon C^{i-1} \to A^{i-1}$. 
    Since $C^*$ is elementwise, $C^i$ is obtained via a pushout
    \begin{center}
    \begin{tikzcd}
    S(p-1) \arrow[d, hook] \arrow[r] 
      & C^{i-1} \arrow[d, hook] \\
    D(p) \arrow[r]
      & C^i
    \end{tikzcd}
    \end{center}
    Let $x$ be the image of a generator under the map $D(p)_p \to C^i_p$. 
    We do case analysis on the value of $r^{i-1}_{p-1} \partial x$.

    \textbf{Case 1:} $r^{i-1}_{p-1} \partial x = 0$. 
    In this case, define $r^i \colon C^i \to A^i$ by the pushout: 
    \begin{center}
    \begin{tikzcd}
    C^{i-1} \arrow[d, hook] \arrow[r,"r^{i-1}"] 
      & A^{i-1} \arrow[d] \\
    C^i \arrow[r,"r^i"]
      & A^i
    \end{tikzcd}
    \end{center}
    By \cref{prop:left_properness} we see that since $r^{i-1}$ is a quasi-isomorphism, $C^{i-1} \to C^i$ is levelwise injective, $r^i$ must also be a quasi-isomorphism. 
    We need to verify that $A^i$ has zero differentials. 
    The only thing to check is that $\partial (r^i_px) = 0$, but this follows by unwinding the definition:
    \[ \partial (r^i_px) = r^i_{p-1} \partial x = r^{i-1}_{p-1} \partial x  = 0\]
    with the last equality given by assumption. 
    In this case, what is happening is that adding in $x$ did not change the rank of $\partial_p$, i.e., $\operatorname{rank}(\partial_p^{i-1}) = \operatorname{rank}(\partial_p^{i})$. 
    Thus we must have that $\dim (\ker(\partial^i_p)) = \dim (\ker(\partial^{i-1}_p)) + 1$. 
    We are therefore increasing the dimension of the $p$-th homology group by 1, so to form $A^i$ from $A^{i-1}$ we must add an element in dimension $p$. 
    The map $r^i$ is then identical to $r^{i-1}$ but sends $x$ to this new element in $A^i_p$. 

    \textbf{Case 2:} $r^{i-1}_{p-1} \partial x \neq 0$. 
    We define $A^i$
    \begin{center}
    \begin{tikzcd}
    S(p-1) \arrow[d, hook] \arrow[r] & C^{i-1} \arrow[d, hook] \arrow[r, "r^{i-1}"] & A^{i-1} \arrow[dd] \\
    D(p) \arrow[r] \arrow[d]           & C^i \arrow[rd, "r^i", dotted]                &                    \\
    0 \arrow[rr]                      &                                              & A^i               
    \end{tikzcd}
    \end{center}
    with $r^i$ given by the induced map.
    To see that $r^i$ is a quasi-isomorphism, we take the pushout in:
    \begin{center}
    \begin{tikzcd}
    C^{i-1} \arrow[r, "r^{i-1}"] \arrow[d] & A^{i-1} \arrow[d] \\
    C^i \arrow[r]       & P             
    \end{tikzcd}
    \end{center} 
    thus writing $r^i$ as a composite of $C^i \to P \to A^i$.
    Since $C^{i-1} \to C^i$ is a monomorphism, \cref{prop:left_properness} tells us that the first factor is a quasi-isomorphism as a pushout of $r^{i-1}$ along a monomorphism.
    Moreover, the pushout pasting lemma implies that
        \begin{center}
    \begin{tikzcd}
    D(p) \arrow[r] \arrow[d] & C^i \arrow[r] & P \ar[d] \\
    0 \arrow[rr]   &    & A^i             
    \end{tikzcd}
    \end{center} 
    is a pushout as well.
    The top map is a monomorphism by assumption, and $D(p) \to 0$ is a quasi-isomorphism, so by \cref{prop:left_properness}, we obtain that $P \to A^i$ is a quasi-isomorphism.
\end{proof}

\begin{remark}
    The requirement that we are over a field is necessary for \cref{prop:reduction}. 
    For example, consider the chain complex
    \[ \cdots \to 0 \to \mathbb{Z} \langle y_1,y_2 \rangle \to \mathbb{Z}\langle x \rangle \]
    where $y_1 \mapsto 2x$ and $y_2 \mapsto x$.
    Suppose we have a filtration on the generators $x \to y_1 \to y_2$. 

    For the theorem to hold, we in particular need a commutative square: 
    \begin{center}
        \begin{tikzcd}
        C^2 \arrow[r] \arrow[d, "r^2", swap] & C^3 \arrow[d, "r^3"] \\
        A^2 \arrow[r]                  & A^3                 
        \end{tikzcd}
    \end{center}
    where each $A^i$ has zero differentials and each $r^i$ is a quasi-isomorphism.
    We have that $A^2$ must be 
    \[\cdots \to 0 \to 0 \to \mathbb{Z}/2\]
    and that $A^3$ must be
    \[\cdots \to 0 \to \mathbb{Z} \to 0\]

    Thus, for such a commutative square to exist, since $r^2(y_1)=0$, we must have that $r^3(y_1)=0$. 
    But we need $r^3$ to be a quasi-isomorphism, so it in particular must map the generator of $H_1(C^3)$ to the generator of $H_1(A^3)$. 
    But the generator of $H_1(C^3)$ is $y_1-2y_2$, and we get that $r^3(y_1-2y_2)=-2r^3(y_2)$.
    Regardless of what $r^3(y_2)$ is, we know that $2r^3(y_2)$ cannot be a generator of $H_1(A^3)$.
\end{remark}

While \cref{prop:reduction} works over an arbitrary field $F$, from this point on we set $F= \bbF 2$.
We begin by fixing notation specific to $\bbF 2$. 

\begin{notation}
For vector spaces, we assume we have a choice of basis.
A vector $x$ is expressed as a set $\{x_1,\dots, x_k\}$, which is the set of basis elements whose coefficient is nonzero when expressing $x$ as a linear combination.
For basis elements, we occasionally write $x_i$ to mean $\{x_i\}$.
Consequently, vector addition $x+y$ is then performed via symmetric difference $\{x_1, \dots,x_k\} \triangle \{y_1,\dots, y_\ell\}$, and instead of $\sum x_i$ we write $\bigtriangle \{x_{i_1},\dots x_{i_k}\}$. 
For ease of future reference, we summarize this translation in the following table:
\begin{center}
   \[ 
    \setlength{\arraycolsep}{25pt}
    \renewcommand{\arraystretch}{1.5}
   \begin{array}{ll}
       
         F &  \bbF 2\\
        \hline
         x \in F^n & \{x_1,\dots, x_k\} \\ 
         0 & \varnothing \\

         x+y & \{x_1,\dots x_k\} \triangle \{y_1 \dots, y_\ell\} \\
 
         \sum_{i \in I} x_i & \bigtriangle_{i \in I} \{x_{i_1}, \dots, x_{i_k}\}
    \end{array}\]
\end{center}
\end{notation}

We now describe how to turn the proof of \cref{prop:reduction} into an algorithm for computing persistent homology of the filtration.
The key insight is that $A^i$ depends only on $A^{i-1}$, $r^{i-1}$, and the element $x$ (along with its boundary $\partial x$) that gets added at time $i$. 
We then compute $r^{i-1} \partial x $, and the different cases tell us how to construct $A^i$ as well as $r^i$. 
After each $i$, we thus only ever need to store $r^i$ and $A^i$. 
We can store $r^i_p$ using a dictionary $R_p$, which has as keys all basis elements $x$ of $C^i_p$ that are births, and as values $R_p[x] = r^i_p(x)$. 
We can also recover the basis of $A^i_p$ from this dictionary, by finding all keys $x$ such that $R_p[x] = \{x\}$. 
At every stage in the algorithm, we update $R_p$ to reflect $r^i_p$.

To compute persistent homology up to dimension $n$, the algorithm starts by initializing empty dictionaries $R_p$ for every $0 \leq p \leq n$. 
For each element $x$ that gets added, we need to first compute $r^{i-1}_{p-1} \partial x$. 
To do this, we initialize a set $r \partial x=\varnothing$ and then for every $y \in \partial x$, we check if $y \in \operatorname{keys}(R_{p-1})$.
If not, this means that $r^{i-1}_{p-1}y = 0$, so we skip it. 
Otherwise, we update $r\partial x$ by adding $R_{p-1}[y]$ over $\mathbb{F}_2$, i.e., taking the symmetric difference $r \partial x \triangle R_{p-1}[y]$.

Once we have computed $r \partial x$, we break into the two cases as in \cref{prop:reduction}. 
Case 1 is if $r \partial x=\varnothing$. 
In this case $x$ is a birth element, and gets added as a basis element into $A_p^i$ with $r_p^ix=x$. 
To signify this, we set $R_p[x] = \{x\}$.

Case 2 is if $r \partial x = \{y_1, \dots , y_k\}$. 
In this case $x$ is a death element which kills the linear combination  $\{ y_1, \dots , y_k\}$, consisting of at least one element, where we assume that $\filt(y_k) > \filt(y_j)$ for all $1 \leq j \leq k-1$.
To get the correct decomposition, we store $(\filt(y_k),\filt(x))$ as a persistence pair of dimension $p-1$.
The vector space $A^i_{p-1}$ is then given by quotienting $\{ y_1, \dots , y_k\}$ to zero, and the new basis is given by deleting $y_k$. 
Thus for any $z$, if $r^{i-1}_{p-1}z = B$, where $y_k \in B$, then $r^{i}_{p-1}z = B \triangle \{y_1, \dots, y_k\}$. 
This replaces (mod 2) $y_k \in B$ with $\{y_1,\dots,y_{k-1}\}$.
To update $R_{p-1}$, we loop though all $z \in \keys(R_{p-1})$ such that $y_k \in R_p[z]$, and set $R_{p-1}[z]=R_{p-1}[z]\triangle \{y_1,\dots, y_k\}$

As stated, the \redzed{} algorithm works for an arbitrary element-wise filtration of chain complexes in $\Ch(\bbF 2)$. 
It assumes the filtration comes with a function \texttt{filtration\_enumerator} which, when called, returns the index of the element in the filtration, the dimension of the element, and a list of the elements in its boundary. 
Pseudocode for the algorithm is given below.
We call this version Naive \redzed, since a few further improvements can be made, which we describe at the end of this section.

\begin{algorithm}[H] 
\caption{Naive \redzed}
\begin{algorithmic}[1]
\Statex \textbf{Input 1:} An elementwise filtration given in the form of \texttt{filtration\_enumerator}
\Statex \textbf{Input 2:} A natural number $n$
\Statex \textbf{Output:} The persistence pairs of the filtration up to dimension $n$
\Statex \hspace*{-\leftmargin}\hrulefill
\State Initialize empty lists $P_0,\dots,P_n$
\State Initialize dictionaries $R_0,\dots ,R_n$

\For{$i=1$ to length of filtration}

\State $x$, $p$, $\partial x$ = \texttt{filtration\_enumerator}$(i)$
\State Initialize $r \partial x = \varnothing$

\For{$y \in \partial x$}

\If{$y \in \keys(R_{p-1})$}

\State $r \partial x = r \partial x \triangle R_{p-1}[y]$

\EndIf

\EndFor

\If{$r \partial x = \varnothing$}

\State $R_{p}[x] = \{x\}$ if $p \leq n$

\Else{}

\State Let $r \partial x = \{y_1,\dots ,y_k\}$ with $\filt(y_k) >\filt(y_j)$ for all $1 \leq j \leq k-1$

\For{$z \in \keys(R_{p-1})$ such that $y_k \in R_{p-1}[z]$}

\State $R_{p-1}[z] = R_{p-1}[z] \triangle r \partial x$

\EndFor

\State Append $(\filt(y_k),i)$ to $P_{p-1}$
\EndIf
\EndFor
\For{$p=0$ to $n$}
\For{$x \in \keys(R_{p})$ such that $R_{p}[x]=\{x\}$}
\State Append $(\filt(x), \infty)$ to $P_p$
\EndFor
\EndFor
\State \Return $P_0, \dots, P_n$
\end{algorithmic}
\end{algorithm}
 Before discussing the correctness of the algorithm, we first consider the following example: 
 \begin{example} \label{ex:redzed}
     Suppose that an elementwise filtration chain complexes arises from an elementwise filtratation of simplicial complexes of dimension at most 2 with four points, $a$, $b$, $c$, and $d$. 

        \begin{center}
        \begin{tikzpicture}[scale=2]

        \coordinate (a) at (0,1);
        \coordinate (b) at (1,1);
        \coordinate (c) at (1,0);
        \coordinate (d) at (0,0);

        \draw (a) -- (b);
        \draw (a) -- (c);
        \draw (a) -- (d);
        \draw (b) -- (c);
        \draw (b) -- (d);
        \draw (c) -- (d);

        \fill (a) circle (1.5pt) node[above left] {$a$};
        \fill (b) circle (1.5pt) node[above right] {$b$};
        \fill (c) circle (1.5pt) node[below right] {$c$};
        \fill (d) circle (1.5pt) node[below left] {$d$};
        \end{tikzpicture}
     \end{center}
     
     The refined linear ordering of the filtration is given by:
     \[a \to b \to c \to d \to ab \to bc \to cd \to ad \to ac \to acd \to abc \to bd \to bcd \to abd\text{.}\]
     In other words, it is a filtration of the boundary of the 3-simplex.
    
     We are interested in computing homology in dimensions 0 and 1. 
     The following table represents the operations performed in the execution of Naive \redzed{} at each step of the filtration, and a detailed description is given below.
    \begin{center}
    \begingroup
    \small
    \setlength{\tabcolsep}{6pt}
    \renewcommand{\arraystretch}{0.98}
    
    \begin{tabular}{|c!{\vrule width 1.2pt}c|c|c|c|c|}
    \hline
    Simplex $x$ & $\partial x$ & $r \partial x$ & $R_0$ & $R_1$ & Pairs\\
    \noalign{\hrule height 1.2pt}
    
    $a$ & $0$ & $0$ & $a \to a$ &  &\\
    \hline
    
    $b$ & $0$ & $0$ &
    \begin{tabular}{@{}c@{}}
    $a \to a$ \\
    $b \to b$
    \end{tabular}
    &  &\\
    \hline
    
    $c$ & $0$ & $0$ &
    \begin{tabular}{@{}c@{}}
    $a \to a$ \\
    $b \to b$ \\ 
    $c \to c$
    \end{tabular}
    &  &\\
    \hline
    
    $d$ & $0$ & $0$ &
    \begin{tabular}{@{}c@{}}
    $a \to a$ \\
    $b \to b$ \\ 
    $c \to c$ \\
    $d \to d$
    \end{tabular}
    &  &\\
    \hline
    
    $ab$ & $a+b$ & $a+b$ &
    \begin{tabular}{@{}c@{}}
    $a \to a$ \\
    $b \to a$ \\ 
    $c \to c$ \\
    $d \to d$
    \end{tabular}
    &  
    & $(b,ab)$\\
    \hline
    
    $bc$ & $b+c$ & $a+c$ &
    \begin{tabular}{@{}c@{}}
    $a \to a$ \\
    $b \to a$ \\ 
    $c \to a$ \\
    $d \to d$
    \end{tabular}
    & 
    & $(b,ab)$, $(c,ac)$\\
    \hline
    
    $cd$ & $c+d$ & $a+d$ &
    \begin{tabular}{@{}c@{}}
    $a \to a$ \\
    $b \to a$ \\ 
    $c \to a$ \\
    $d \to a$
    \end{tabular}
    &  
    & 
    \begin{tabular}{@{}c@{}}
    $(b,ab)$, $(c,ac)$, $(d,cd)$
    \end{tabular} \\
    \hline
    
    $ad$ & $a+d$ & $a+a=0$ &
    \begin{tabular}{@{}c@{}}
    $a \to a$ \\
    $b \to a$ \\ 
    $c \to a$ \\
    $d \to a$
    \end{tabular}
    &  $ad \to ad$ 
    & 
    \begin{tabular}{@{}c@{}}
    $(b,ab)$, $(c,ac)$, $(d,cd)$
    \end{tabular}\\
    \hline
    
    $ac$ & $a+c$ & $a+a=0$ &
    \begin{tabular}{@{}c@{}}
    $a \to a$ \\
    $b \to a$ \\ 
    $c \to a$ \\
    $d \to a$
    \end{tabular}
    &  
    \begin{tabular}{@{}c@{}}
    $ad \to ad$ \\
    $ac \to ac$
    \end{tabular}
    & 
    \begin{tabular}{@{}c@{}}
    $(b,ab)$, $(c,ac)$, $(d,cd)$
    \end{tabular}
    \\
    \hline
    
    $acd$ & $bd+ad+ac$ & $ad+ac$ &
    \begin{tabular}{@{}c@{}}
    $a \to a$ \\
    $b \to a$ \\ 
    $c \to a$ \\
    $d \to a$
    \end{tabular}
    &  
    \begin{tabular}{@{}c@{}}
    $ad \to ad$ \\
    $ac \to ad$
    \end{tabular}
    &
    \begin{tabular}{@{}c@{}}
    $(b,ab)$, $(c,ac)$, $(d,cd)$,\\
    $(ac,acd)$
    \end{tabular}\\
    \hline
    
    $abc$ & $bc+ac+ab$ & $ad$ &
    \begin{tabular}{@{}c@{}}
    $a \to a$ \\
    $b \to a$ \\ 
    $c \to a$ \\
    $d \to a$
    \end{tabular}
    &   
    &
    \begin{tabular}{@{}c@{}}
    $(b,ab)$, $(c,ac)$, $(d,cd)$,\\
    $(ac,acd)$, $(ad,abc)$
    \end{tabular}\\
    \hline
    
    $bd$ & $b+d$ & $a+a=0$ &
    \begin{tabular}{@{}c@{}}
    $a \to a$ \\
    $b \to a$ \\ 
    $c \to a$ \\
    $d \to a$
    \end{tabular}
    &   
    \begin{tabular}{@{}c@{}}
    $bd \to bd$
    \end{tabular}
    & 
    \begin{tabular}{@{}c@{}}
    $(b,ab)$, $(c,ac)$, $(d,cd)$,\\
    $(ac,acd)$, $(ad,abc)$
    \end{tabular}\\
    \hline
    
    $bcd$ & $cd+bd+bc$ & $bd$ &
    \begin{tabular}{@{}c@{}}
    $a \to a$ \\
    $b \to a$ \\ 
    $c \to a$ \\
    $d \to a$
    \end{tabular}
    &   
    &
    \begin{tabular}{@{}c@{}}
    $(b,ab)$, $(c,ac)$, $(d,cd)$,\\
    $(ac,acd)$, $(ad,abc)$, $(bd,bcd)$
    \end{tabular}\\
    \hline
    
    $abd$ & $bd+ad+ab$ & $0$ &
    \begin{tabular}{@{}c@{}}
    $a \to a$ \\
    $b \to a$ \\ 
    $c \to a$ \\
    $d \to a$
    \end{tabular}
    &   
    &
    \begin{tabular}{@{}c@{}}
    $(b,ab)$, $(c,ac)$, $(d,cd)$,\\
    $(ac,acd)$, $(ad,abc)$, $(bd,bcd)$
    \end{tabular}\\
    \hline
    
    \end{tabular}
    
    \endgroup
    \end{center}
     The last step in the algorithm is to add an interval $(a, \infty)$ since it is the only remaining value in $R_0$. 

     Each of the first four steps adds a single vertex, which becomes a generator in $H_0$, and thus $R_0[x]=\{x\}$ in each case. 
     When $ab$ is added, we have $\partial (ab) = r \partial (ab) = a+b$. 
     Since $b$ is younger, we set $R_0[b]=\{a\}$ and add the pair $(b,ab)$. 
     When $bc$ is added, we have $\partial (bc) = b+c$, but since $R_0[b]=\{a\}$, we get $r \partial bc=a+c$. 
     Since $c$ is younger, we set $R_0[c]=a$ and add $(b,bc)$ to the set of pairs.
     A similar process is repeated for $cd$.

     When $ad$ is added, we have $\partial (ad) = a+d$, and since $R_0[d]=\{a\}$, $r \partial (ad)=a+a=0$. 
     Thus $ad$ is a birth, and we set $R_1[ad]=\{ad\}$. 
     When $ac$ is added, it is also a birth, and we set $R_1[ac]=\{ac\}$. 
     When $acd$ is added, we have that $\partial (acd) = bd+ad+ac$. 
     We have that $R_1[ac]=\{ac\}$, $R_1[ad]=\{ad\}$, and $bd \notin \keys(R_1)$, so $r \partial abd = ad+ac$. 
     We thus set $R_1[ac]=\{ac\} \triangle \{ad,ac\}=\{ad\}$, and add $(ac,acd)$ as a pair. 
     When $abc$ is added we get that $\partial (abc) = bc+ac+ab$, and thus $r \partial (abc) = R_1[ac]=ad$. 
     We thus update $R_1[ac]=\{ad\} \triangle \{ad\} = \varnothing$, so we delete $R_1[ac]$. 
     Identically we delete $R_1[ad]$, and add $(ad,abc)$ as a pair. 
     The remaining steps in the algorithm are all similar cases to the steps above.  
 \end{example}

We are now ready to state and prove correctness of Naive \redzed.

\begin{theorem} \label{TH:redzed_correct}
    Naive \redzed{} correctly computes all persistence pairs for an elementwise filtration $C^*$. 
\end{theorem}

The proof requires the following lemma:

\begin{lemma} \label{lemma:R_p correct}
    For an elementwise filtration, after every step $i$ in the filtration and for every $0 \leq p \leq n$, $R_p$ correctly encodes $r^i_p$, in that for every basis element $x \in C^i_p$: 
    $$r^i_p(x) = \begin{cases}
        R_p[x] & \text{if } x \in \keys(R_p) \\
        0                                   & \text{if } x \notin \keys(R_p)
    \end{cases}$$
\end{lemma}
\begin{proof}
We proceed by induction. 
    The base case is immediate, since $C^0_p = 0 = A^0_p$ and $R_p$ and $R_p^{-1}$ are empty.
    Now suppose that for all $p$, the dictionary $R_p$ correctly encodes $r^i_p$. 
    Suppose at time $i$ an element of dimension $p$ is added. 
    We have two cases:

    \textbf{\underline{Case 1:}} $r^{i-1}_{p-1} \partial x=0$. 
    In this case, we know that we land in case 1 of the algorithm since $r \partial x$ computed by the algorithm is equal to the symmetric difference of $R_{p-1}[y]$ over all $y \in \partial x$. This in turn is equal to the addition mod $2$ of $r^{i-1}_{p-1}y$ over all $y \in \partial x$ since by assumption $R_{p-1}$ correctly encodes $r^{i-1}_{p-1}$. 
    The element $y \in \partial x$ that are not keys in $R_{p-1}$ are such that $r^{i-1}_{p-1}y=0$.

    By \cref{prop:reduction}, the vector space $A^i$ is obtained by adding $x$ as a new basis element to $A^{i-1}_p$ and setting $r^i_p(x)=x$. 
    This is exactly what the algorithm does, by setting $R_p[x]=\{x\}$. 
    Since $r^i_{p}$ remains the same on all other basis elements, and we did not change $R_{p}$ aside from adding $x$, we are done.

    \textbf{\underline{Case 2:}} $r^{i-1}_{p-1}x \neq 0$. 
    Again by assumption since $R_{p-1}$ encodes $r^{i-1}_{p-1}$, we get that 
    $r\partial x$ computed by the algorithm is equal to $r^{i-1}_{p-1}x$, so we land in case 2 of the algorithm. 
    Suppose $r \partial x = \{y_1,\dots,y_k\}$ for basis elements $y_1, \ldots, y_k$, where $y_k$ is the youngest element in the set.
    In this case, only $A^{i-1}_{p-1}$ is changed, with $r \partial x = \{y_1,\dots,y_k\}$ quotiented to zero, and $y_k$ removed from the basis. 
    Suppose $z$ is a basis element of $C^{i-1}_{p-1}$ such that $r^{i-1}_{p-1} = B$, where $y_k \in B$. 
    Then $r^i_{p-1}(z) = B \triangle \{y_1,\dots,y_k\}$.
    Thus for every $z$ such that $y_k \in R_{p-1}[z]$, we need to replace $y_k$ with $y_1,\dots ,y_{k-1}$, which is exactly done by lines 15-17 in the algorithm.
    Thus after all $z \in \keys(R_{p-1})$ with $y_k \in R_{p-1}[z]$ are processed, $R_{p-1}$ correctly encodes $r^i_{p-1}$. 
    Since only dimension $p-1$ was changed, this still holds for all other dimensions. 
\end{proof}

We now prove \cref{TH:redzed_correct}.

\begin{proof}[Proof of \cref{TH:redzed_correct}]
    Since each $r_i$ is a quasi-isomorphism and each $A^i$ has zero differentials, in each dimension $p$ the sequence 
    \[A^0_p \to A^1_p \to \dots \to A^n_p\]
    is the persistence module associated with the $p$-th homology of $C^*$. 
    We thus need to show that the decomposition given by \redzed{} is a decomposition of this persistence module. 

    For $i<j$, we want to understand the transition maps $A^i \to A^j$, which we will denote $s_i^j$. 
    It suffices to determine it degreewise on basis elements. 
    Suppose $x \in A^{i}_p$ is a basis element, i.e., a key in $R_p$ such that after time $i$, $R_p[x]=\{x\}$ . 
    Then we can view $x$ as an element of $C^i_p$ such that $r^i_p(x)=x$. 
    By commutativity of the diagram
    \begin{center}
    \begin{tikzcd}
    C^i \arrow[d, "r^i",swap] \arrow[r, hook] & C^j \arrow[d, "r^j"] \\
    A^i \arrow[r, "s_i^j"]               & A^j                 
    \end{tikzcd}
    \end{center}
    we get that $s^i_j(x) = r^j_p(x)$. 
    After step $j$ of the algorithm, by \cref{lemma:R_p correct}, we thus have that:
    \[s^i_j(x) = \begin{cases}
         R_p[x] & x\in \keys(R_p)\\
        0 & x\notin \keys(R_p)\text{.}
        \end{cases}\]

    Let $B^*$ be the persistence module given by the algorithm. 
    That is, $B^* = \bigoplus_I [b_i,d_i)$, where $\{(b_i,d_i)\}_I$ are the pairs given by the algorithm. 
    We want to show that $B^* \cong A^*$. 
    Note that since $B^*$ may contain infinite intervals, we are viewing $A^*$ as infinite as well, being the identity at each step after the final filtration time.

    We can choose a basis for $B^i$ to be the set of all $x$, such that $x$ is a birth element corresponding to an interval $(j,k)$ where $j \leq i <k$. 
    Since the basis for $A^i$ is exactly that set of elements, we see that $A^i \cong B^i$ for all $i$. 
    The natural map exhibiting this isomorphism need not be an isomorphism $A^* \cong B^*$, however, as it may fail the necessary commutativity conditions. 

    Instead, we need to choose a new basis for $A^i$. 
    We do this by replacing each basis element $x$ with a new basis element as follows. 
    Suppose $x$ is a basis element. 
    If $s_i^jx=x$ for all $j \geq i$, then we do not replace $x$, and define the value of the map $A^i \to B^i$ at $x$ to be $x$.
    Otherwise, suppose $j$ is minimal such that $s_i^jx = \{y_1, \dots, y_k\}$ for $x \notin \{y_1, \dots, y_k\}$. 
    Then replace $x$ in the basis with $\{y_1, \dots, y_k , x\}$. 
    Note that by construction, $x$ is younger than all of $y_1 , \dots ,y_k$, so after we do this we still maintain a basis.
    For the same reason, this is also well defined, since each $y_1,\dots, y_k$ being older than $x$ ensures they are present at time $i$.
    Then define $A^i \to B^i$ to send $\{y_1, \dots, y_k,  x\}$ to $x$. 
    Since this map is an isomorphism on basis elements under this new basis, it is still an isomorphism.
    It suffices to show that for all $i$ the square 
    \begin{center}
        \begin{tikzcd}
        A^i  \arrow[d] \arrow[r] & A^{i+1} \arrow[d] \\
        B^i \arrow[r]            & B^{i+1}          
        \end{tikzcd}
    \end{center}
    commutes. 

    Since $C^*$ is an elementwise filtration, either a new basis element gets added into $A^{i+1}$ or a generator dies. 
    If a new basis element $x$ gets added in at time $i+1$, then $A^{i+1}\to B^{i+1}$ will be the same map as $A^i \to B^i$ on all the basis elements of $A^i$, so we are done. 
    Now suppose $s_i^{i+1}x = \{y_1, \dots, y_k\}$. 
    Then $\{y_1, \dots, y_k, x\}$ is a basis element of $A_i$ that gets mapped to $x \in B^i$. 
    Moreover, an element $z$ must have been added at time $i+1$ such that $r \partial z = \{y_1, \dots, y_k, x\}$. 
    Thus $(\filt(x),i+1)$ gets stored as a pair by the algorithm, and thus $x$ gets mapped to $0$ by $B^i \to B^{i+1}$. 
    We also have that $$s_i^{i+1}\{y_1, \dots, y_k,x\} = \{y_1,\dots,y_k, x\}\triangle \{y_1, \dots, y_k, x \}=\varnothing$$
    So the square commutes for this basis element. 
    Note that this is the only basis element containing $x$, since $x$ cannot appear in $r \partial z$ for any future $z$. 
    Thus for all other basis elements, $A^i \to A^{i+1}$ is the identity, and hence the square commutes.
\end{proof}

\begin{remark} \label{rmk:works_over_field}
    The \redzed{} algorithm can be modified to work over an arbitrary field $F$. 
    To do this, the values of the dictionaries $R_p$ need to be weighted sets, with the weights being scalars in $F$. 
    For $r \partial x = a_1y_1+\dots +a_ky_k$, instead of taking the symmetric difference $R_p[z] \triangle \{y_1,\dots,y_k\}$ for all $z \in R^{-1}_p[y_k]$, one must subtract the weighted set $\{a_1y_1,\dots , a_ky_k\}$ from $R_p[z]$ with a sufficient coefficient to cancel the $y_k$ term in $R_p[z]$. 
\end{remark}

Finally, we comment on the relation between Naive \redzed{} and the algorithm found in \cite{Dey-Fan-Wang}.

\begin{remark} \label{rmk:dey-fan-wang}
For elementwise filtrations of chain complexes arising from filtrations of simplicial complexes, our maps $r^i_p\colon C^i_p \to A^i_p$ recover the notion of an \emph{annotation} $a \colon K(p) \to \mathbb{F}_2^g$ used in \cite{Dey-Fan-Wang}.
Specifically, if $K^i$ is the $i$-th simplicial complex in a filtration and $C(K^i)$ is the chain complex generated by simplices, the map $r^i_p \colon C_p(K^i) \to A_i$ can be viewed as an annotation $K(p) \to A_i \cong \mathbb{F}_2^g$, where $K^i(p)$ is the set of $p$-simplices in $K^i$. 
In particular, this annotation is valid in the sense of \cite[Def. 3.2]{Dey-Fan-Wang}, and is the annotation used in \cite[\S4.1]{Dey-Fan-Wang} for elementary inclusions.
For a simplex $\sigma$ added at time $i$, our cases of $r \partial \sigma=0$ and $r \partial \sigma\neq0$ are the same as cases $a \partial \sigma = 0 $ and $a \partial \sigma \neq 0$ of \cite{Dey-Fan-Wang}. 
While the exact implementation differs, Naive \redzed{}'s updates to $R_p$ are analogous to the updates performed in \cite{Dey-Fan-Wang} to the annotation. 
In this sense, \redzed{} can be viewed as a different perspective on and a generalization of a part of the \cite{Dey-Fan-Wang} algorithm. 
\end{remark}

\subsection{Matrix interpretation}

The formulation of Naive \redzed{} given above conceals its relation to the persistence pairing algorithm.
To elucidate this relation, we will translate Naive \redzed{} into the matrix-theoretic language and explain the specific speedups uncovered by this translation.

To begin, let us briefly review the standard persistence pairing algorithm, described in greater detail in \cref{sec:prelims}.
It works with the filtration boundary matrix $M$, and applies a reduction algorithm, reducing columns left to right. 
    To reduce a given column with index $M_x$, columns $M_y$ with $\filt(y)<\filt(x)$ are added to $M_x$, until it either has a pivot distinct from all previous columns or is zero. 
    After this reduction is done for all columns, the persistence pairs can be read off from the pivot entries. 
    A pivot entry $(x,y)$ in corresponds to a persistence pair $(\filt(x),\filt(y))$ of dimension $\dim(x)$, and tells us that the $M_x$ column of gets reduced to zero. 
    The infinite intervals may be read off by finding all zero columns $M_x$ that do not appear in any of the pairs $(x,y)$.
    If $x$ is such an element, then $(\filt(x), \infty)$ is an infinite interval in the diagram.

    Of relevance to us are three speedups that can be applied to the algorithm, namely, compression, exhaustive reduction, and retrospective reduction.
    These were studied in detail in \cite{keeping-it-sparse} and we refer interested readers to that reference for a comprehensive overview.

    1.~\emph{Compression} \cite{clear-and-compress} is the process of deleting rows corresponding to death elements once they are identified.
    These rows cannot affect the reduction and can therefore be omitted. 

    2.~\emph{Exhaustive reduction} \cite{edelsbrunner-olsbock,edelsbrunner-zomorodian,zomorodian-carlsson:computing-persistent-homology} reduces each column even after finding a unique pivot by adding previous columns that have pivot entries in positions where the current column has nonzero entries. 
    Its purpose is to sprasify the matrix during reduction.

    3.~\emph{Retrospective reduction} \cite{keeping-it-sparse} uses the current column to reduce all previous columns via right-to-left addition and thus zeroes the entire row to the left of the pivot. 
    (This is a slightly different formulation than that in \cite{keeping-it-sparse}, but is the one relevant to our algorithm.) 
    Retrospective reduction occasionally refers to the combination of retrospective and exhaustive reduction; however, we refer to the two processes separately. 
    Like exhaustive reduction, retrospective reduction has also been used to sparsify the matrix during reduction. 

    To translate between the language of \redzed{} and the persistence pairing algorithm, we first fix some notation. 
\begin{notation}
    We write $y \in M_x$ to mean that $M_x$ has a nonzero entry in the row corresponding to $y$. 
    If $y$ is a birth element that has been killed, this corresponds to a unique column, denoted $M^y$, that has pivot $y$ (again, meaning the pivot entry is in the row corresponding to $y$).
\end{notation}
    
    The main connection between \redzed{} and the standard algorithm comes from the following proposition:

    \begin{proposition} \label{prop:matrix}
        Let $C^*$ be an elementwise filtration, and $M$ the corresponding filtration boundary matrix. 
        Suppose an element $x$ is added at time $i$. 
        Let $M_x$ be the reduced column in $M$ after step $i$ using compression, and exhaustive and retrospective reductions.
        Then $r \partial x$ computed by Naive \redzed{} is exactly the set of nonzero entries in $M_x$.
    \end{proposition}

    \begin{proof}
        We prove this using induction; however, we require a stronger inductive hypothesis. 
        Namely, we prove that at a step $i$:
        \begin{enumerate}
            \item $r \partial x$ consists of the nonzero entries in $M_x$;
            \item for an element $y$, we have $R_p[y]=\{y\}$ if and only if $y$ is a birth element that does not appear as a pivot of any earlier column $M_z$;
            \item for an element $y$, we have $y \notin R_p[y] = \{z_1, \dots, z_k\}$ if and only if $y$ appears as a pivot of a column $M^y$.
            In this case $\{z_1, \dots, z_k\}$ is the set of nonzero entries that lie above $y$ in $M^y$.
        \end{enumerate}

        The base case is empty.
        The proof in the inductive step proceeds as follows:
        \begin{itemize}
            \item prove 1.~assuming that 2.~and 3.~are true prior to the $i$-th step;
            \item prove 2.~and 3.~hold after the $i$-th step.
        \end{itemize}
    
        If $x$ is a birth element, by correctness of Naive \redzed, we get that $r \partial x=0$, and by correctness of the pairing algorithm $M_x$ gets reduced to zero. 
        In this case, retrospective reduction is not applied, and in Naive \redzed, we set $R_p[x]=\{x\}$ and change nothing else, so we are done.

        Now suppose $x$ is not a birth element.
        We want to show that $r \partial x$ in Naive \redzed{} is the set of all nonzero elements in reduced column $M_x$ after exhaustive reduction and compression. 
        
        The first step in reduction is compression, so we first zero all death entries in $M_x$. 
        In Naive \redzed, this is done implicitly since death entries are not added into $R_p$, so they are ignored in the computation of $r \partial x$.
        We may assume death elements are processed first, so after this step $r \partial x$ is still empty, and the elements in $\partial x$ that have yet to be processed are the nonzero entries that remain in column $M_x$ after compression. 
        We thus may assume that both $\partial x$ and $M_x$ contain no death entries.

        Since all columns to the left of $M_x$ have been exhaustively and retrospectively reduced, they contain no entries that are pivots of any other columns. 
        To reduce exhaustively reduce $M_x$, we thus find all elements $y \in M_x$ that are pivots of previous columns, and cancel them by taking $M_x=M_x\triangle M^y$.
       
        Define
        \[
        I_x = \{y \in \keys{R_p} \ | \ y \notin R_p[y] \} \cap \partial x 
        \qquad \text{and} \qquad 
        J_x = \{ y \in \keys(R_p) \ | \ R_p[y]=\{y\} \} \cap \partial x \text{.}
        \]
        By the inductive hypothesis, $I_x$ is the set of $y \in \partial x$ that appear as pivots of previous columns, and $J_x$ is the set of $y \in \partial x$ that are births that do not appear as previous pivots. 
        Since we are assuming $\partial x$ contains no deaths, $\partial x = I_x \cup J_x$. 
        Writing $M_x'$ for the column $M_x$ after reduction, we have that (note that the $M_x$ on the left is the reduced column and $M_x$ on the right is unreduced):
        \begin{center}
        \begin{align*}
            M_x' & = (\bigtriangle_{y \in I_x} M^y) \triangle M_x \\
            &= (\bigtriangle_{y \in I_x} R_p[y] \triangle \{y\}) \triangle \partial x  && \text{since } M^y = R_p[y] \triangle \{y\} \text{ by I.H.} \\
            &= (\bigtriangle_{y \in I_x} R_p[y]) \triangle I_x \triangle \partial x \\
            &= (\bigtriangle_{y \in I_x} R_p[y]) J_x && \text{since } J_x = \partial x \setminus I_x\\ 
            &= (\bigtriangle_{y \in \partial x} R_p[y])  && \text{since } R_p[y] = \{y\} \text{ for } y \in J_x \text{ and } \partial x = I_x \cup J_x\\
            &= r \partial x
        \end{align*}
        \end{center}
        Thus $r \partial x$ is the set of nonzero entries in $M_x$ after exhaustive and retrospective reduction. 

        It remains to show that after completion of the $i$-th step, 2.~and 3.~still hold.
        Since $x$ is a death, we have that $r \partial x = \{y_1,\dots ,y_k\}$, so $y_k$ is the pivot of $M_x$. 
        For every $z$ with $\filt(z) \leq \filt(x)$, if $y_k \in M_z$, we cancel $y_k$ by taking $M_z = M_z \triangle M_x = M_z \triangle r\partial x$.
        if $w$ is the pivot of $M_z$, then we have $y_k \in R_p[w]$ by the inductive hypothesis.
        In Naive \redzed, we set $R_p[w] = R_p[w] \triangle r \partial x$ in line 16. 
        This is exactly the set of nonzero entries above $w$ in $M^w=M_z$ after taking $M_z = M_z \triangle r \partial x = (R_p[w] \triangle \{w\}) \triangle r \partial x$.
        Thus, condition 3.~is still satisfied. 
        For condition 2., note that $y_k$ cannot appear in $R_p[z]$ for any birth $z$ that is still alive (not a pivot of any column), so we are done.
    \end{proof}

    With this perspective, Naive \redzed{} can be viewed as a different implementation of compression with exhaustive and retrospective reduction. 
    Our rephrasing allows for a natural further improvement, which we explore in \cref{sec:active}.
    We first describe two minor improvements to this version of \redzed, which are both more computational than mathematical.

    \subsection{Further improvements}
    The first change to Naive \redzed{} is very minor, meant to save memory and make the forthcoming improvements of \cref{sec:active} more natural. 
    The idea is that once a birth element $x$ is such that $R_p[x] = \varnothing$, then $R_p[x]$ will stay empty for the remainder of the algorithm.
    We can therefore delete the key from the dictionary $R_p$, and it will not change any of the computations in the algorithm. 

    The second change comes from the observation that in case 2 of the algorithm, we spend a lot of time looping through $z \in \keys(R_p)$ to check if $y_k \in R_p[z]$, namely line 15 of Algorithm 1.
    To avoid this loop, we maintain a reverse dictionary $R_p^{-1}$, which has as keys basis elements of $A^i_p$, and values $R_p^{-1}[x]=\{y \in \keys(R_p) \ | \ x \in R_p[y] \}$.
    Then, instead of looping through $z \in \keys{R_p}$ such that $y_k \in R_p[z]$, we loop through $z \in R_p^{-1}[y_k]$.
    To maintain $R_p^{-1}$, it starts as empty, and we update it every time a new element $x$ is added.
    In case 1, where $r\partial x=0$, since we set $R_p[x]=\{x\}$, we need to also set $R_p^{-1}[x]=\{x\}$.
    In case 2, we have that $r \partial x=\{y_1,\dots,y_k\}$. 
    In this case, we normally loop though all $z \in \keys(R_{p-1})$ and if $y_k \in R_{p-1}[z]$ take $R_{p-1}[z] = R_{p-1}[z] \triangle \{y_1,\dots,y_k\}$, but instead we loop though $z \in R_{p-1}^{-1}[y_k]$ and for each $z$ take $R_{p-1}[z]=R_{p-1}[z]\triangle \{y_1,\dots,y_k\}$.
    We then need to update $R_p^{-1}[y_i]$ for each $i$. 
    For each $z \in R_p[y_k]$, if beforehand we had that $y_i \in R_p[z]$, now $y_i \notin R_p[z]$, and if beforehand $y_i \notin R_p[z]$, now $y_i \in R_p[z]$.
    Thus to update $R_p^{-1}[y_i]$, we can take $R_p^{-1}[y_i]=R_p^{-1}[y_i] \triangle R_p^{-1}[y_k]$ for each $i$.
    We also then delete $R_p^{-1}[y_k]$ from the dictionary. 
    Full psudeocode with these changes is given below: 

    \begin{algorithm}[H] 
\caption{\redzed}
\begin{algorithmic}[1]
\Statex \textbf{Input 1:} An elementwise filtration given in the form of \texttt{filtration\_enumerator}
\Statex \textbf{Input 2:} A natural number $n$
\Statex \textbf{Output:} The persistence pairs of the filtration
\Statex \hspace*{-\leftmargin}\hrulefill
\State Initialize empty lists $P_0,\dots,P_n$
\State Initialize dictionaries $R_0,\dots ,R_n$, and $R_0^{-1}, \dots, R_n^{-1}$

\For{$i=1$ to length of filtration}

\State $x$, $p$, $\partial x$ = \texttt{filtration\_enumerator}$(i)$
\State Initialize empty set $r\partial x$

\For{$y \in \partial x$}

\If{$y \in \keys(R_{p-1})$}

\State $r\partial x = r\partial x \triangle R_{p- 1}[y] $

\EndIf

\EndFor

\If{$r\partial x = \varnothing$}

\State $R_{p}[x] = \{x\}$ if $p \leq n$

\Else{}

\State Let $r\partial x = \{y_1,\dots ,y_k\}$ with $\filt(y_k) >\filt(y_j)$ for all $1 \leq j \leq k-1$

\For{$z \in R_{p-1}^{-1}[y_k]$}

\State $R_{p-1}[z] = R_{p-1}[z] \triangle r\partial x$

\If{$R_{p-1}[z] = \varnothing$}

\State Delete $R_{p-1}[z]$ from the dictionary

\EndIf

\EndFor

\For{$j=1$ to $k-1$}
\State $R_{p-1}^{-1}[y_j] = R_{p-1}^{-1}[y_j] \triangle R_{p-1}^{-1}[y_k]$
\EndFor

\State Delete $R_{p-1}^{-1}[y_k]$ from the dictionary 
\State Append $(\filt(y_k),i)$ to $P_{p-1}$
\EndIf
\EndFor
\For{$m=0$ to $n$}
\For{$x \in \keys(R_{p}^{-1})$}
\State Append $(\filt(x), \infty)$ to $P_m$
\EndFor
\EndFor
\State \Return $P_0, \dots, P_n$
\end{algorithmic}
\end{algorithm}

We shall refer to this improved algorithm as \redzed.
\section{Active enumeration} \label{sec:active}
While the \redzed{} algorithm is stated for arbitrary filtrations of chain complexes, our primary motivation is to improve computation times for Vietoris--Rips filtrations. 
In this setting, \redzed{} allows for a natural further improvement, which we call \textit{active enumeration}.
In the case of Vietoris--Rips filtrations with active enumeration employed, we call our algorithm \redzedvr. 

When implemented directly, the persistence pairing algorithm suffers one major bottleneck: when computing $n$-th persistent homology, almost all time is spent reducing $(n+1)$-birth columns \cite{silva-morozov-vejdemo-johansson:dualities,bauer:ripser}. 
Since these columns correspond to generators in dimension $n+1$, they are irrelevant to the $n$-th homology computation, and thus in practice one wants a way to avoid reducing them. 
In Ripser \cite{bauer:ripser}, this is done using a combination of birth clearing and cohomology. 
These same tricks are not compatible with \redzed, however, so we require a new method to avoid reducing $(n+1)$-births.
We begin by introducing a bit of terminology.

\begin{definition}
    Let $X=(X,d)$ be a set equipped with a distance function, and $\VR(X)$ the corresponding (refined) Vietoris--Rips filtration. 
    Let $i \geq 0$, and $\sigma$ a simplex such that $\filt(\sigma) \leq i$. 
    We say that $\sigma$ is \textit{active} at time $i$ if $r^i \sigma \neq 0$, otherwise we say $\sigma$ is \textit{inactive}.
\end{definition}

Let us make a few comments about this definition.
First, we note that all death simplices are automatically inactive.
For birth simplices, the two classifications: alive--vs--dead and active--vs--inactive do not coincide.
All inactive simplices are dead and all alive simplices are automatically active, but a simplex can be dead and active, when its homology class was merged with nonzero homology class of another (still alive) simplex.
At a given index $i$, the active simplices in dimension $p$ can be read off from $\keys(R_p)$. 
Finally, note that all birth simplices start off active, and once they are inactive, they remain inactive. 

Let us now explain how the active--vs--inactive distinction could be used to ``weed out'' some of the $(n+1)$-birth simplices.
Proceeding naively, after adding an $n$-simplex $\sigma$, we would look for the next simplices in the filtration, which must be those $(n+1)$-simplices $\tau$ that contain $\sigma$ as their youngest face.
As indicated above, we wish to avoid enumerating those that are births.
If an $(n+1)$-simplex $\tau$ has only inactive faces, then $r \partial \tau=0$, so $\tau$ is an $(n+1)$-birth simplex and we do not need to include it in the filtration.
Put differently, we only need to enumerate the $(n+1)$-simplices with at least one active face.
\iffalse
After adding an $n$-simplex $\sigma$, we would normally search for all $(n+1)$-simplices that have $\sigma$ as their youngest face and process them. 
During this process, it is possible that $\sigma$ becomes inactive. 
If this is the case, then certain $n+1$ simplices that are enumerated may have all their faces be inactive (and thus be a birth). 
We want to avoid enumerating such simplices.
\fi

In searching for such $(n+1)$-simplices, no checks are required after a death $n$-simplex:

\begin{proposition}
    Let $(X,d)$ be a set equipped with a distance function. 
    Suppose a simplex $\sigma$ of dimension $n$ is added at time $i$ in the (refined) Vietoris--Rips filtration of $(X,d)$ is a death simplex. 
    Then $\sigma$ does not appear as the youngest face of any $(n+1)$-dimensional simplices. 
\end{proposition}

\begin{proof}
    Suppose $\sigma$ is the youngest face of $\tau$, and $\tau$ is the eldest such coface of $\sigma$ with this property. 
    Then $(\sigma,\tau)$ form an apparent pair in the sense of \cite{bauer:ripser}, which contradicts $\sigma$ being a death simplex.
\end{proof}
Thus, after each birth, simplex $\sigma$ of dimension $n$, we search for $(n+1)$-simplices that have at least one active face, and have $\sigma$ as their youngest face.
We do this in two steps: 
\begin{itemize}
    \item \textbf{Step 1:} search for all $(n+1)$-simplices that have $\sigma$ as their youngest face, along with another active face $\tau$; and
    \item \textbf{Step 2:} if $\sigma$ is still active after the conclusion of Step 1, we search for one extra $(n+1)$-simplex that has $\sigma$ as its face, and no other active faces.
\end{itemize}

\textbf{Step 1:} We loop through $\tau \neq \sigma \in \keys(R_n)$ and check for each $\tau$, if there is an $(n+1)$-simplex with $\sigma$ as its youngest face, that has both $\sigma$ and $\tau$ as faces.
To do this, we first check that $\sigma$ and $\tau$ share a face, which happens if $|\sigma \cap \tau| = n$. 
If so, there is an $(n+1)$-simplex that has $\sigma$ and $\tau$ as faces, but we must ensure that $\sigma$ is the youngest face of the new simplex.
Let $v = \sigma \setminus \tau$ and $w = \tau \setminus \sigma$. 
We first check if $d(v,w) \leq \filt(\sigma)$, which if so, ensures that $\filt(\sigma \cup \tau) = \filt(\sigma)$, which must be the case for $\sigma$ to be the youngest face. 
While this check is not necessary, it is faster than checking directly that $\sigma$ is the youngest face, so it weeds out cases where we can avoid checking.
If $d(v,w) \leq \filt(\sigma)$, we then directly check that $\sigma$ is the youngest face, and if so, we then process the simplex $\sigma \cup \tau$. 
To avoid enumerating a simplex twice, we temporarily store a list of the $w$'s that we have checked.
Once we have looped through all $\tau \in \keys(R_n)$, we have enumerated all $(n+1)$-simplices that have both $\sigma$ and at least one other active face in their boundary.
If $\sigma$ is no longer active, then all other $(n+1)$-simplices will have no active faces, so we are done. 

\textbf{Step 2:} If $\sigma$ is still active, we only need to search for one other $(n+1)$-simplex.
To do this, we loop through all vertices that are not in the stored list of checked $w$'s, and see if they can be added to $\sigma$ to form an $(n+1)$-simplex.
Once one is found, call it $\tau$, we process it and immediately stop searching for more.
This is because since $\sigma$ is the only active face, we will have that $\overline{\partial}\tau = r\sigma$ (note that we are omitting the index from $r$ here). 
Thus, when we quotient $\overline{\partial}\tau =0$, we will get that $r\sigma=0$, so $\sigma$ will become inactive. 
Thus, all future $(n+1)$-simplices created by $\sigma$ will have no active faces.

\begin{example}
    We now revisit \cref{ex:redzed} with the addition of active enumeration. 
    Since we are concerned with first homology, when a birth $1$-simplex is added, we check for $2$-simplices with active faces using step 1 and then step 2.
    The following table shows how all $2$-simplices necessary for the computation of the first homology group are found via active enumeration.
    The columns contain in order: the $1$-simplex $x$ added, the value of $r \partial x$, whether active enumeration applies, the values stored in $R_1$, and the $2$-simplices found in each step.
    \begin{center}
    \begin{tabular}{|c!{\vrule width 1.2pt}c|c|c|c|c|}
    \hline
    Simplex $x$ & $r \partial x$ & A.E.? & $R_1$ & Step 1 found & Step 2 found\\
    \noalign{\hrule height 1.2pt}

    $ab$ & $a+b$ & No &  &  & \\
    \hline 

    $bc$ & $a+c$ & No &  &  & \\
    \hline 

    $cd$ & $a+d$ & No &  &  & \\
    \hline 

    $ad$ & $0$ & Yes & 
    \begin{tabular}{@{}c@{}}
    $ad \to ad$
    \end{tabular}
    &  & \\
    \hline 

    $ac$ & $0$ & Yes & 
    \begin{tabular}{@{}c@{}}
    $ad \to ad$ \\
    $ac \to ac$ \\
    \end{tabular}
    & $acd$ & $abc$ \\
    \hline 

    $bd$ & $0$ & Yes & 
    \begin{tabular}{@{}c@{}}
    $bd \to bd$
    \end{tabular}
    &  & $bcd$\\
    \hline

    \end{tabular}
    \end{center}
    Simplices $ab$, $bc$, and $cd$ are death simplices, and therefore not subject to active enumeration.
    When $ad$ is added, since it is a birth, we search for $2$-simplices created by $ad$. 
    Since there are no active simplices that are not $ad$, we skip step $1$. 
    In step $2$, no $w$ works to form a 2-simplex, so none are found. 
    After $ac$ is added, we again search for $2$-simplices.
    This time $ad$ is an active face, so we check if $ac \cup ad = acd$ forms a valid 2-simplex created by $ac$, which it does, so it gets processed. 
    After $acd$ gets processed, we get that $R_1[ac]=R_1[ad]=ad$, so $ac$ is still active.
    In this case we proced to step 2, where we find $abc$, which kills $ad$ and makes both $ac$ and $ad$ inactive. 
    When $bd$ gets added, we again skip step 1 since there are no other active 1-simplices. 
    In step 2, we find $bcd$, which kills $bd$, and immediately exit the process. 
    In particular, we avoid enumerating $abd$.
    Depending on the order we searched, we could have found $abd$ first, and avoided enumerating $bcd$.   
\end{example}

Note that this does not guarantee that all $n+1$ births will not be enumerated, as it is possible for an $(n+1)$-simplex to have multiple active faces and still be a birth. 
For example, if an $(n+1)$-simplex $\omega$ has only active simplices $\sigma$ and $\tau$, and $r \sigma = r \tau$, then $r \partial \omega = r \sigma + r \tau =0$. 
This does, however, avoid enumerating \emph{almost all} $(n+1)$-births. 
For example, for a dataset of size 400, for $n=1$ there are generally 10,507,399 dimension 2 birth simplices, and in the worst case we tested, we enumerate 296,911 of them, meaning we avoid enumerating over $97\%$ of all $(n+1)$-dimensional birth simplices. 
This worst case generally happens when all points are concentrated around one homology feature. 
For example, this data was taken to be 400 points about a unit circle, with slight noise added in. 
In other cases, the number is generally much better than $97\%$. 
In fact, in all other cases we tested, we avoid enumerating more than $99.95\%$ of the $n+1$ birth simplices, and in some cases, this number is actually $100\%$. 
To see the exact types of datasets we tested this on, see \cref{sec:experiments}.

\iffalse
While this improvement, as stated, is specific to Vietoris--Rips filtrations, a similar idea could hold in any case where it is easy to enumerate only future elements that have at least one active face.
\fi

\subsection{Matrix interpretation}
Similar to \redzed, we can translate these ideas into the matrix-theoretic language. 
Specifically, we translate what it means for a simplex to be active, as well as why simplices with all inactive faces will get reduced to zero.

Recall from the proof of \cref{prop:matrix} that for a birth simplex $\sigma$, if $\sigma \notin R_p[\sigma]$, then $R_p[\sigma]$ is equal to the set of nonzero entries above $\sigma$ in $M^\sigma$, the column with $\sigma$ as its pivot.
If this set is empty, then $R_p[\sigma]=\varnothing$, so we remove it from the dictionary, and $\sigma$ becomes inactive.

Thus, after the $i$-th step of the filtration, the active simplices are exactly the simplices that are either current unpaired birth simplices, or birth simplices $\sigma$ that have been killed, but have nonzero entries above $\sigma$ in $M^\sigma$. 
Once the column $M^\sigma$ has been maximally reduced, meaning the only nonzero entry is $\sigma$, then $\sigma$ becomes inactive. 

For a matrix interpretation of why $(n+1)$-simplices with inactive faces need not be included in the filtration, suppose an $(n+1)$-simplex $\tau$ is such that at the time it is added, all of its faces are inactive simplices. 
Let $M_\tau$ be the column corresponding to $\tau$.
Since all faces are inactive, this means that all the faces of $\tau$ are either death simplices or ones that appear as a pivot in a column with no other nonzero entries. 
If this is the case, we know the $M_\tau$ will get reduced to zero. 
This is because any face that is a death simplex will be cleared by the compression algorithm, so that will already be zero in $M_\tau$. For any face $\sigma$ that is an inactive birth simplex, we can immediately cancel it by adding $M^\sigma$ to $M_\tau$.
Since the rest of $M^\sigma$ is zero, this does not change the rest of $M_\tau$.
Once we do this for all of its faces that are inactive birth simplices, $M_\tau$ will be reduced to zero. 

\begin{example}
Consider the matrix below, and suppose we have reduced the matrix up until the $d$-th column.
    \begin{center}
\begin{tikzpicture}[x=0.8cm, y=0.5cm]

% entries
\node at (0, 0) {};
\node at (1, 0) {$a$};
\node at (2, 0) {$b$};
\node at (3, 0) {$c$};
\node at (4, 0) {$d$};
\node at (5, 0) {$e$};

\node at (0,-1) {$f$};
\node at (1,-1) {$0$};
\node at (2,-1) {$0$};
\node at (3,-1) {$\mathbf{1}$};
\node at (4,-1) {$0$};
\node at (5,-1) {$0$};

\node at (0,-2) {$g$};
\node at (1,-2) {$\mathbf{1}$};
\node at (2,-2) {$0$};
\node at (3,-2) {$0$};
\node at (4,-2) {$\mathbf{1}$};
\node at (5,-2) {$\mathbf{1}$};

\node at (0,-3) {$h$};
\node at (1,-3) {$0$};
\node at (2,-3) {$0$};
\node at (3,-3) {$\mathbf{1}$};
\node at (4,-3) {$0$};
\node at (5,-3) {$0$};

\node at (0,-4) {$i$};
\node at (1,-4) {$\mathbf{1}$};
\node at (2,-4) {$0$};
\node at (3,-4) {$0$};
\node at (4,-4) {$\mathbf{1}$};
\node at (5,-4) {$\mathbf{1}$};

\node at (0,-5) {$j$};
\node at (1,-5) {$0$};
\node at (2,-5) {$\mathbf{1}$};
\node at (3,-5) {$0$};
\node at (4,-5) {$\mathbf{1}$};
\node at (5,-5) {$0$};

\node at (0,-6) {$k$};
\node at (1,-6) {$0$};
\node at (2,-6) {$0$};
\node at (3,-6) {$0$};
\node at (4,-6) {$0$};
\node at (5,-6) {$\mathbf{1}$};

% table lines
\draw (-0.5,-0.5) -- (5.5,-0.5);
\draw (0.5,0.5) -- (0.5,-6.5);

% box highlighting column d
\draw[thick, rounded corners=1pt]
    (3.5,0.5) rectangle (4.5,-6.5);

% strikethrough row g
\draw[thick]
    (-0.5,-2) -- (5.5,-2);

\end{tikzpicture}
\end{center}
The crossed out row corresponds to that row being compressed, so it is ignored.
The only active simplex that is not alive is $h$, since it appears as the pivot of column $c$, which has another nonzero entry in $f$. 
The simplices $k$ and $f$ are also active, since they correspond to births not yet killed.
We know this because their rows have not been compressed and they are not pivots of any column.
Note that $i$ is not active, since the other nonzero entry in column $a$ is zeroed by compression. 
The column $d$ only has entries in rows $i$, $j$, and $g$. 
Since $g$ gets compressed the only remaining entries are in $i$ and $j$, which are both inactive. 
Thus $d$ would get reduced to zero simply by adding columns $a$ and $b$. 
Since $d$ has no active faces, however, the active enumeration shortcut would avoid enumerating and processing column $d$ altogether. 
\end{example}

This highlights the fact that the combination of exhaustive and retrospective reduction, along with compression, is necessary for the active enumeration shortcut.
Without both exhaustive and retrospective reduction, we wouldn't be reducing all columns as much as possible, and thus wouldn't be able to track active simplices.
Without compression, very few simplices would get flagged as inactive.
Since death simplices can never be pivots, they are impossible to directly reduce.
Thus simplices with at least one death simplex in their boundary may never get flagged as inactive without compression.

\section{Implementation} \label{sec:implementation}
%In this section, we go over the implementation of \redzedvr{}. 
Our implementation of \redzedvr{} was first written in Julia, and then translated to C++, with both versions available at \url{https://github.com/nkershaw01/RedZeD}.
Many parts of the implementation are inspired by Ripser, including the maintenance of neighborhoods for each point, the usage of the combinatorial numbering system (although we do not make full use of it as in Ripser), and the enclosing radius. 
Bauer notes in \cite{bauer:ripser} that the enclosing-radius optimization was first proposed by Henselman-Petrusek \cite{henselman-petrusek2017matroids}.

\paragraph{Input and output.} 
The inputs to \redzedvr{} are a Float32 distance matrix $M$, and a homological degree $n$. 
The matrix can be lower triangular, or square, in which case it will be converted to lower triangular.
%The entries of the matrix are Float32.
An optional input is the threshold, in which the filtration is enumerated only up until a specified radius.
If no threshold is given, it defaults to the \textit{enclosing radius}, defined as $$R_e=\min_{x \in X} \max_{y \in X} d(x,y)$$
after which the simplicial complex is a cone, so all homology in non-negative degrees will be trivial. 
If a threshold is specified, it uses $\min(\text{threshold},R_e)$.
In C++, the output is an array $[P_0,P_1,\dots,P_n]$ of persistence pairs, where $P_i$ is an array of pairs in degree $i$. 
In Julia, the output is a dictionary with keys a dimension $d$ whose value is $P_d$. 

Another optional input is \texttt{simplex\_pairs}, which defaults to false and, when set to \texttt{true}, returns simplex pairs alongside the intervals. 
This has the same format as the intervals, such that $\texttt{spx\_pairs}[d][k]$ is the corresponding birth-death simplex pair to the interval $\texttt{intervals}[d][k]$.
Simplices are stored as vertices, numbered from zero in matrix order.
Birth-death pairs are tuples $(b,d)$, where $b$ is the birth simplex and $d$ is the death simplex, with $d$ an empty tuple for infinite intervals.

\paragraph{Filtration enumeration.} Rather than enumerate the filtration in strict order, we enumerate and process dimension by dimension, in increasing order. 
The exception is dimension $n$, which, due to active enumeration, is processed interleaved with the $(n+1)$-simplices active enumeration finds. 

For dimension 1, we build and sort all edges at once, and then process them sequentially.
We sort edges via iterated bucket sort, first distributing them into $\min(E/4,2^{16})$ buckets based on their radius, where $E$ is the total number of edges in the filtration, ignoring the threshold. 
Both numbers were chosen empirically, with a cap of $2^{16}$ used for memory, which has the added benefit of speed due to allowing variables such as the counts in each bucket to stay in cache.
If a bucket is of size $>32$, we again use bucket sort on it, dividing it into ($\#$ of edges in the bucket)$/4$ further buckets. 
This process is iterated, sorting buckets directly when their size reaches $\leq 32$, when the depth reaches 5, or when all edges in a bucket are the same radius.
When directly sorting, we use insertion sort when the size $\leq 32$, and comparison sort otherwise. 
Radius ties are broken via the combinatorial index.
The function $\texttt{build\_edges}$ implements this sorting procedure.

The natural approach for dimensions $\geq 2$ is to enumerate using the same method; however, this has the drawback of needing to simultaneously store all the simplices in that dimension, which becomes the dominant memory cost. 
Instead, we generate $d$-dimensional simplices for $d \geq 2$ one edge at a time: given an edge $e$, we find all $d$-dimensional simplices with max edge $e$ using a method similar to Ripser's coface enumeration.
For each dimension $d$, we iterate the edges in order while maintaining a bitset for each vertex $v$, with 1 in position $w$ if $v \sim w$.
Enumerating $ d$-dimensional simplices with max edge $e$ becomes a problem of iteratively intersecting neighbourhoods, which we can do in decreasing vertex order so that generated simplices are done in order of their combinatorial index (which is colexicographical on vertices), which avoids sorting them once they are generated.
Simplices are processed immediately after they are found, meaning we do not have to store them. 
This was observed to be faster, but more importantly, around 16 times more memory efficient for $H_2$ and higher as opposed to repeating the strategy for edges.
This is done in $\texttt{generate\_stream}$. 

\paragraph{Simplex storage.} 
Once a simplex is generated and processed, it is referenced using its \textit{rank}, its position in the filtration for its dimension, rather than its combinatorial index as in Ripser. 
The reason for this is that whenever we calculate $r \partial \sigma \neq 0$, we need to find the youngest element. 
If we stored simplices via their combinatorial index, this would require decoding the index to vertices, and then performing $d(d+1)/2$ matrix reads to compute the diameter for each element in $r \partial \sigma$.
Using rank allows us to choose the largest element in $r \partial \sigma$, which is already sorted (see the paragraph on the symmetric difference). 
For each live simplex, we also store its diameter, used for reporting persistence pairs.
Storing the diameters would avoid the $d(d+1)/2$ matrix reads if the combinatorial index were to be used; however, it would still require accessing memory for each simplex in $r \partial \sigma$ and finding the maximum diameter. 

\paragraph{Symmetric difference.} 
To compute symmetric differences, one of the primary operations of the algorithm, we store sets as sorted arrays.
To take the symmetric difference of two sets, start at position 1 of each array, advance whichever entry is smaller, and add it to the output list. 
If at any point the entries are equal, advance both and add neither to the output. 
When one array is exhausted, append the remaining entries to the output. 
This method has the added benefit of making pivot selection free, done by simply reading off the last element in the array. 

\paragraph{Dictionaries.} 
Since a large portion of the algorithm is spent accessing and editing dictionaries, how they are implemented is important. 
The main structure we use for the $R$ and $R^{-1}$ dictionaries is \texttt{SetDict}, which is built on the structure \texttt{FastMap}, an open addressing hash table with linear probing. 
It uses backwards-shift deletion instead of tombstones, which would otherwise accumulate as simplices are frequently added and deleted as keys, increasing both memory use and running time.
The input to \texttt{FastMap} is the rank of a simplex, and the output is a pool slot, or $0$ if the rank is not a key. 
The pool slot corresponds to a pool in \texttt{SetDict}, which stores arrays \texttt{sets}, \texttt{sing}, \texttt{sdiam}, and \texttt{svert}. The array \texttt{sets} stores the $R$-values (or $R^{-1}$ values), $\texttt{sdiam}$ stores simplex diameters for $R^{-1}$ (used to return a persistence pair), and $\texttt{svert}$ stores simplex vertices for $R$.
The array \texttt{sing} is used for $R$, and stores $\texttt{sing}[\text{slot}] = 0$ if $\texttt{sets}[\text{slot}]$ is not a singleton, and $\texttt{sing}[\text{slot}] = v$ if $\texttt{sets}[\text{slot}]=\{v\}$. 
This is used for an active enumeration shortcut, described in the eponymous paragraph.

Since faces arrive as vertex lists and not simplex ranks, we need a way to convert a list of vertices to a pool slot. 
This is needed for all faces up to dimension $n-1$, as dimension $n$ faces are handled separately by active enumeration. 
To do this, we use an additional structure \texttt{FaceSlot} to map into the $R$-pool slots. 
We first convert the vertex set to its combinatorial index and pass it to \texttt{FaceSlot}, which uses either a dense or a sparse method to return the $R$-pool slot.
The dense method stores an array \texttt{dense}, which is of size the number of possible $d$-simplices for dimension $d$, and is such that $\texttt{dense}[i]$ is the slot for the simplex with combinatorial index $i$, or $0$ if it is not a key of $R_d$. 
The sparse method uses \texttt{FastMap} to do the same thing.
The reason why a dense method can work here is that this only needs to be stored up to dimension $n-1$, so there will be far fewer potential keys than $R$, which needs to be used up to dimension $n$. 
We switch from dense to sparse when there are more than $2^{24}$ possible keys, which is mainly useful for large datasets called with a low threshold. 

\paragraph{Active enumeration.}
Recall that active enumeration is done in two steps: step 1 searches for $(n+1)$-simplices $\tau$ that have $\sigma$ as a face and additionally at least one other active face, and step 2 tries to find one additional $\tau$ not seen by step 1. 
Step 1 must first find all active $n$-simplices that share a face with $\sigma$. 
The most obvious approach is to loop through keys of $R_n$, and for every key, check if it shares a face with $\sigma$. 
Using this method is slow, however, because there are generally a large number of active simplices and the face-sharing check fails the vast majority of the time. 
Instead, we store necessary additional information in \texttt{ActiveReg}. 
It stores is a \texttt{FaceSlot} table that takes in a combinatorial index of an $(n-1)$-simplex and points at an \texttt{ActiveReg} slot. 
The slot then gives access to the ranks of all active $n$-simplices containing that face, the extra vertex in the active $n$-simplex, and its $R$-pool slot. 

Thus, given a birth $n$-simplex $\sigma$, we first compute its faces and use \texttt{ActiveReg} to list all active simplices containing each face.
For each active $\tau$ sharing a face with $\sigma$, let $w$ be the extra vertex in $\tau$ and $v$ the extra vertex in $\sigma$. 
We first check that $d(v,w)\leq \filt(\sigma)$ using a similar bitset to that used for filtration enumeration, and if so add $w$ to a list of candidate vertices if it wasn't already present. 
We also add the slot of $\tau$ to a list \texttt{seedlists}$[w]$, so that by the end of this process \texttt{seedlists}$[w]$ stores the pool slots of all active faces in $\sigma \cup \{w\}$ other than $\sigma$ itself.
This process is done by \texttt{collect\_candidates}.

After all candidate vertices have been found, we loop through the list to check whether $\sigma$ is the youngest face of each corresponding coface.
Prior to this check, which requires computing the combinatorial index of each face and potentially checking several diameters, we implement the shortcut mentioned in the paragraph on dictionaries. 
Although active enumeration avoids enumerating most $(n+1)$-births, some are still enumerated.
Among the births that do get processed, the vast majority involve cancelling singletons, as when $\sigma \cup \{w\}$ has active faces $\sigma$ and $\tau$ with $R[\sigma] = R[\tau] = \{\omega\}$.
Recall that \texttt{SetDict} stores an array \texttt{sing}, such that \texttt{sing}[slot] records whether or not \texttt{sets}[slot] is a singleton, and if so stores the single vertex. 
This allows us to quickly check if a candidate would fall under this case, and if so, immediately skip it. 
We perform this check before checking if $\sigma$ is the youngest face of $\sigma \cup \{w\}$, because it is by far the faster check. 
An alternate shortcut would be to simply compute $r \partial (\sigma \cup \{w\})$ beforehand; however, this this was slower in our experiments. 
This is because the singleton case catches almost all potential births anyway, and is cheaper because it avoids taking symmetric differences and requires access to an array with a smaller memory footprint. 

Step 2 of active enumeration is implemented by \texttt{has\_leftover\_coface}, which intersects the neighbourhoods of all vertices $v \in \sigma$ and removes the vertices seen in step 1.
For every remaining vertex $w$, it checks if $\sigma$ is the youngest face of $\sigma \cup \{w\}$ and if so returns \texttt{true} immediately. 
If no coface is found, it returns \texttt{false}. 

\section{Experiments} \label{sec:experiments}
As stated, the primary motivation for the development of \redzed{} is to improve computation times and memory usage for Vietoris--Rips persistent homology, for which we now have \redzedvr. 
To our knowledge, the high water mark in this enterprise belongs to the algorithm Ripser \cite{bauer:ripser}.
In this section, we compare computation times and memory with Ripser for various datasets.

We compare the C++ version of \redzedvr{} with the C++ implementation of Ripser, version 1.2.1. 
All experiments were done on a laptop with 64 GB of RAM and an Intel i9-13900HX processor.
For a fair comparison, since \redzedvr{} takes input a distance matrix that is also the format we input into Ripser.

We begin by testing on curated datasets, where we can control the characteristics in order to isolate different factors.
We also test on some of the common benchmark datasets that are used in \cite{bauer:ripser}.
We generate datasets in four ways with different characteristics. 
The first is $\texttt{random\_euclidean}(n,\dim)$. 
This simply generates $n$ random points inside $[0,1]^{\dim}$. 
Since these points are random in a cube in $\mathbb{R}^{\dim}$, there is very little topological signal. 
The second method is $\texttt{noisy\_sphere}(n,\dim,\sigma)$. 
This randomly generates $n$ points in $\mathbb{R}^{\dim+1}$ about the unit sphere $S^{\dim}$, and adds noise according to the standard deviation $\sigma$. 
The third method is $\texttt{grid\_spheres}(n,m,\dim,\sigma)$. 
This creates  $n$ unit spheres $S^{\dim}$ of $m$ points each in $\mathbb{R}^{\dim +1}$, and stacks them in a $\dim+1$-dimensional grid. 
Again, it adds noise according to the standard deviation $\sigma$. 
Finally, $\texttt{random\_distance}(n)$ just returns a random $n \times n$ distance matrix, where the diagonal is zero, and it is symmetric. 
Each entry is chosen at random to be between 0 and 1.
While this has no true topological signal, since the dataset is only a semimetric, it creates a large amount of signal due to noise.

We test $\texttt{random\_distance}(n)$ in homology degrees 1, 2, and 3. 
For $H_1$, we test on the range $(50,450)$, for $H_2$ we test on the range $(40,175)$ and for $H_3$ we use $(30,125)$. 
We also test $\texttt{noisy\_sphere}(n,\dim,\sigma)$ in degrees 1, 2, and 3. 
For $H_1$, we set $\dim=1$ and let $n$ vary in $(200,3800)$. 
For $H_2$ we set $\dim=2$ and let $n$ vary in $(80,900)$, and for $H_3$ we set $\dim=3$ and let $n$ vary in $(50,430)$. 
For $\texttt{grid\_spheres}(n,m,\dim,\sigma)$ we only test in degrees 1 and 2, since in $H_3$ we would not be able to increase $n$ to any substantial number. 
For $H_1$ we set $m=20$, $\dim=1$, and vary $n$ between 5 and 250 (100 and 5000 total points). 
For $H_2$ we set $m=80$, $\dim=2$, and vary $n$ between 2 and 21 (160 and 1680 total points). 
For $\texttt{random\_euclidean}(n,\dim)$, we fix $\dim=2$ and test $H_1$ for $n$ varying in $(400,12000)$. 
We also fix $\dim=10$ and test degrees 1, 2 and 3. 
For $H_1$ we let $n$ vary in $(200,6000)$, for $H_2$ we use $(80,1800)$, and for $H_3$ we use $(40,330)$. 
For each experiment, we report the runtime and peak memory usage (peak RSS).
The results can be seen below in \cref{fig:time} and \cref{fig:mem} (note the log scales on both the x and y axes).

\begin{figure}[H]
    \centering
    \includegraphics[width=\textwidth]{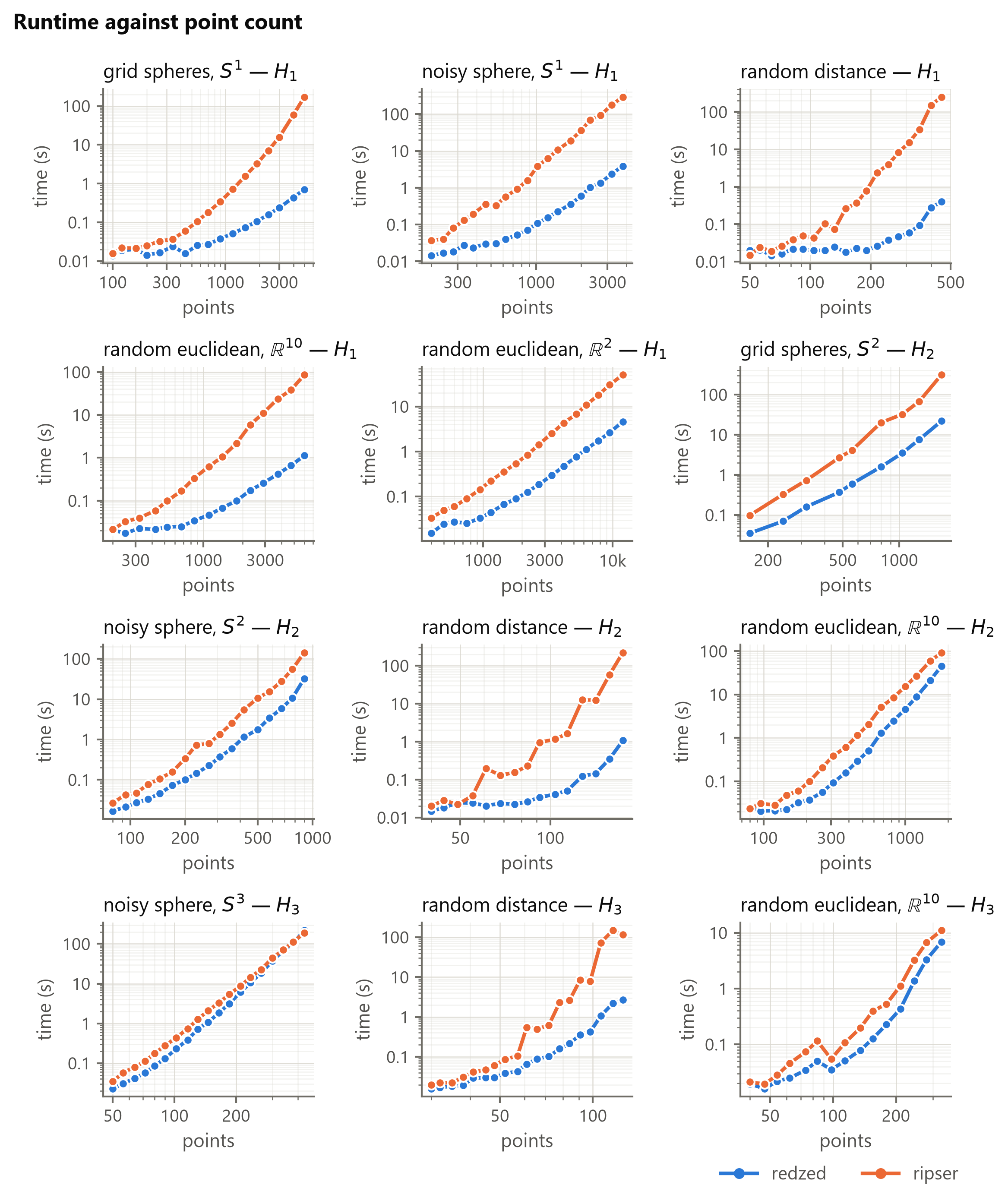}
    \caption{Ripser vs \redzedvr{} time and memory for various datasets}
    \label{fig:time}
\end{figure}

\begin{figure}[H]
    \centering
    \includegraphics[width=\textwidth]{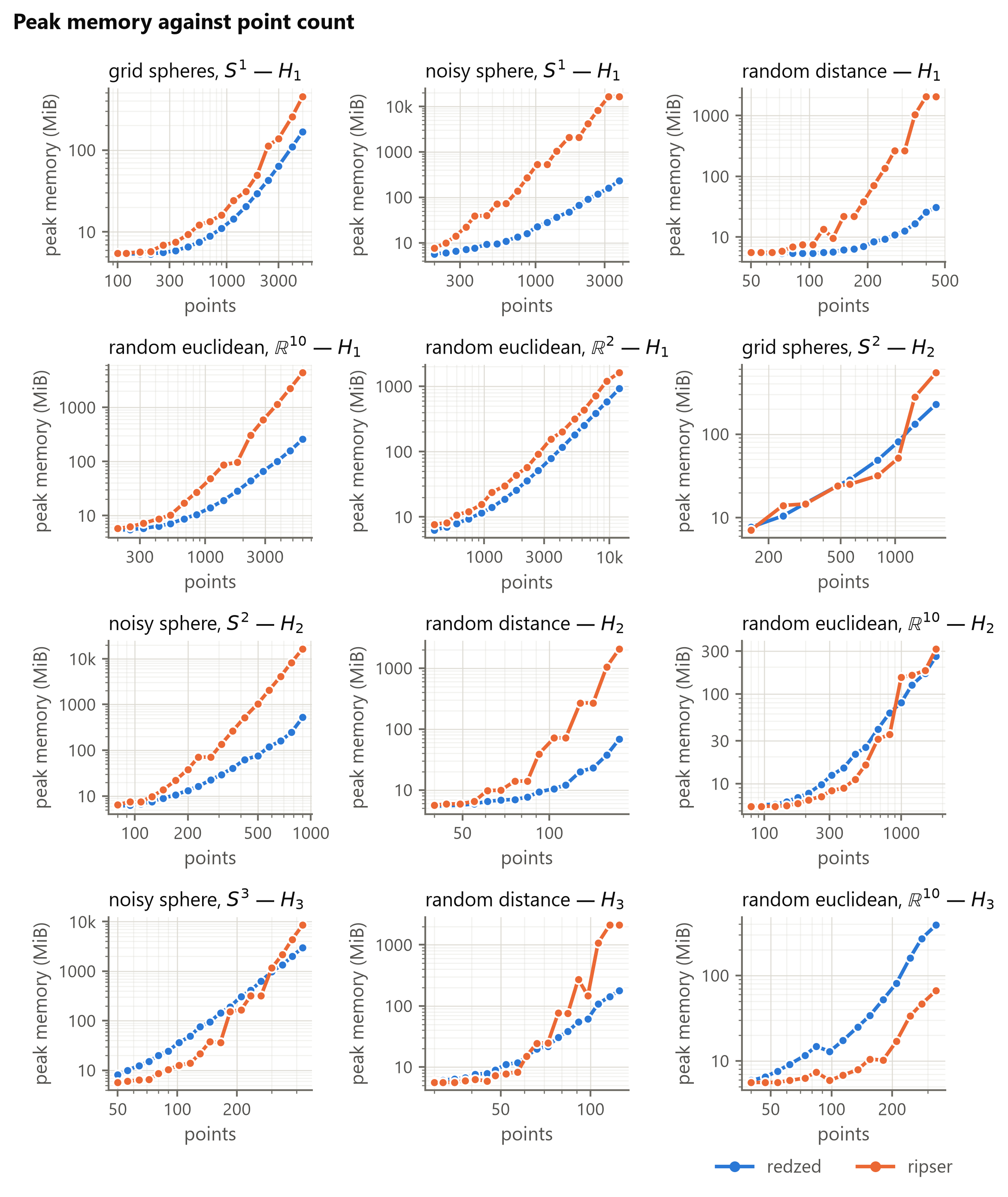}
    \caption{Ripser vs \redzedvr{} time for various datasets}
    \label{fig:mem}
\end{figure}

We also run on several of the fixed datasets tested in \cite{bauer:ripser}.
Namely, we test on \texttt{sphere\_3} in degrees 1, 2 and 3, \texttt{fract\_r} in degrees 1 and 2, \texttt{dragon} in degree 1, \texttt{o3\_1024} in degree 1 and degree 3 with a threshold of 1.8, and finally \texttt{o3\_4096} in degree 1 and in degree 3 with a threshold of 1.4. 
The results of those tests, along with the results of the largest runs for the above datasets, can be found below.
\begin{table}[htbp] \label{table:timemem}
  \centering
  \setlength{\tabcolsep}{4pt}
  \begin{tabular}{llrrrrrrr}
    \toprule
    & & & \multicolumn{3}{c}{Time (s)} & \multicolumn{3}{c}{Peak memory (MiB)} \\
    \cmidrule(lr){4-6} \cmidrule(lr){7-9}
    Dataset & & pts & \redzedvr{} & Ripser & Speedup & \redzedvr{} & Ripser & Ratio \\
    \midrule
    grid spheres, $S^{1}$ & $H_1$ & 5,000 & 0.706 & 171 & 242 & 167 & 452 & 2.71 \\
    noisy sphere, $S^{1}$ & $H_1$ & 3,800 & 3.87 & 292 & 75.4 & 231 & 16,500 & 71.4 \\
    random distance & $H_1$ & 450 & 0.404 & 251 & 621 & 30.9 & 2,070 & 66.9 \\
    random euclidean, $\mathbb{R}^{10}$ & $H_1$ & 6,000 & 1.12 & 86.5 & 76.9 & 258 & 4,400 & 17.1 \\
    random euclidean, $\mathbb{R}^{2}$ & $H_1$ & 12,000 & 4.62 & 50.7 & 11.0 & 921 & 1,620 & 1.76 \\
    grid spheres, $S^{2}$ & $H_2$ & 1,680 & 21.9 & 313 & 14.3 & 228 & 546 & 2.39 \\
    noisy sphere, $S^{2}$ & $H_2$ & 900 & 32.5 & 143 & 4.40 & 522 & 16,400 & 31.4 \\
    random distance & $H_2$ & 175 & 1.07 & 221 & 207 & 68.9 & 2,080 & 30.3 \\
    random euclidean, $\mathbb{R}^{10}$ & $H_2$ & 1,800 & 45.4 & 92.7 & 2.04 & 264 & 319 & 1.21 \\
    noisy sphere, $S^{3}$ & $H_3$ & 430 & 219 & 189 & 0.866 & 3,000 & 8,540 & 2.84 \\
    random distance & $H_3$ & 125 & 2.67 & 116 & 43.6 & 178 & 2,110 & 11.8 \\
    random euclidean, $\mathbb{R}^{10}$ & $H_3$ & 330 & 6.87 & 11.0 & 1.61 & 390 & 66.5 & 0.171 \\
    \midrule
    \texttt{sphere3} & $H_1$ & 192 & 0.0167 & 0.0240 & 1.43 & 5.40 & 6.10 & 1.13 \\
    \texttt{sphere3} & $H_2$ & 192 & 0.154 & 0.441 & 2.86 & 16.7 & 70.2 & 4.20 \\
    \texttt{sphere3} & $H_3$ & 192 & 4.24 & 6.58 & 1.55 & 246 & 86.3 & 0.351 \\
    \texttt{fract\_r} & $H_1$ & 512 & 0.0235 & 0.120 & 5.08 & 7.50 & 16.6 & 2.21 \\
    \texttt{fract\_r} & $H_2$ & 512 & 1.29 & 3.01 & 2.34 & 33.0 & 18.0 & 0.545 \\
    \texttt{o3\_1024} & $H_1$ & 1,024 & 0.0599 & 0.898 & 15.0 & 16.4 & 145 & 8.82 \\
    \texttt{o3\_1024} (thr.\ 1.8) & $H_3$ & 1,024 & 0.491 & 1.36 & 2.76 & 81.5 & 18.2 & 0.223 \\
    \texttt{dragon} & $H_1$ & 2,000 & 0.114 & 0.969 & 8.51 & 30.0 & 98.4 & 3.28 \\
    \texttt{o3\_4096} & $H_1$ & 4,096 & 1.40 & 53.3 & 38.0 & 177 & 8,360 & 47.2 \\
    \texttt{o3\_4096} (thr.\ 1.4) & $H_3$ & 4,096 & 15.8 & 26.1 & 1.65 & 1,070 & 179 & 0.168 \\
    \bottomrule
  \end{tabular}
  \caption{Time and memory comparisons between \redzedvr{} and Ripser for various datasets}
\end{table}

\redzedvr{} is faster than Ripser in all datasets tested except for noisy sphere $S^3$ for $n=430$, where it takes 219 seconds compared to Ripser's 189. 
\redzedvr{} is also more memory efficient in 16/22 of the datasets. 

\redzedvr{} appears to perform best in homology degree 1, where it is both faster and more memory efficient in all cases tested. 
For sufficiently large datasets (ones not affected by startup time), the speedups range from 8.5 to 621 times faster. 
For speed, the improvements are most significant for high-dimensional data, (76.9 times faster for $\texttt{random\_euclidean}(6000,10)$) and data with more topological signal (242, 75.4, and 621 times faster for $\texttt{grid\_spheres}(250,20,1,0.05)$, $\texttt{noisy\_sphere}(3800,1,0.05)$ and $\texttt{random\_distance}(450)$ respectively). 
For memory the biggest improvement is $\texttt{noisy\_sphere}(3800,1,0.05)$, where \redzedvr{} used 321 MiB compared to Ripser's 16,500.
The other $H_1$ datasets ranged from 1.13 to 66.9 times less memory usage. 

The $H_2$ and $H_3$ numbers are not as pronounced, but \redzedvr{} remains generally faster than Ripser. 
For $H_2$ on metric data, speedups range from 2.04 times faster to 14.3, and for the non-metric case \redzedvr{} was 207 times faster.
The only loss on memory for $H_2$ was \texttt{fract\_r}, where \redzedvr{}  was at 33.0 MiB compared to 18.0 MiB for Ripser. 
For $H_3$ \redzedvr{} remains faster for all datasets aside from \texttt{noisy\_sphere}. 
The two non-thresholded memory losses in $H_3$ are $\texttt{sphere3}$ and $\texttt{random\_euclidean}$, the two cases that lack real $H_3$ signal. 
For the two $H_3$ thresholded cases, \redzedvr{} is faster, but loses quite significantly on memory. 

\redzedvr{} also appears to scale better than Ripser in many cases, especially in $H_1$. 
For \texttt{grid\_spheres} $H_1$, at 1500 points \redzedvr{} is 22 times faster than Ripser. At 3000 points this number is 58, and at 5000 points it is 242. 
In degree 1, we see similar results for time for \texttt{noisy\_sphere} (17 at $n=750$, 60 at $n=2000$, and 75 at $n=3800$), \texttt{random\_distance} (15 at $n=150$, 176 at $n=275$, and 621 at $n=450$) and \texttt{random\_euclidean} $\mathbb{R}^{10}$ (15 for $n=1400$, 43 for $n=2900$, and 77 for $n=6000$).
Memory also appears to be scaling better in degree 1 for \texttt{noisy\_sphere}, \texttt{random\_euclidean} $\mathbb{R}^{10}$, and \texttt{random\_distance}.

To highlight the power of active enumeration, we also test \redzed{}, the version without active enumeration, on a few datasets: 

\begin{table}[htbp]
  \centering
  \setlength{\tabcolsep}{4pt}
  \label{tab:ablation}
  \begin{tabular}{llrrrrrrr}
    \toprule
    & & & \multicolumn{3}{c}{Time (s)} & \multicolumn{3}{c}{Peak memory (MiB)} \\
    \cmidrule(lr){4-6} \cmidrule(lr){7-9}
    Dataset & & $n$ & \redzedvr{} & Ripser & \redzed{} & \redzedvr{} & Ripser & \redzed{} \\
    \midrule
    grid spheres, $S^{1}$ & $H_1$ & 560 & 0.0256 & 0.104 & 4.76 & 7.50 & 12.2 & 2,140 \\
    noisy sphere, $S^{1}$ & $H_1$ & 540 & 0.0306 & 0.330 & 7.12 & 9.60 & 72.0 & 3,490 \\
    random distance & $H_1$ & 450 & 0.404 & 251 & 5.62 & 30.9 & 2,070 & 2,320 \\
    random euclidean, $\mathbb{R}^{10}$ & $H_1$ & 525 & 0.0243 & 0.0985 & 2.89 & 7.10 & 10.2 & 1,620 \\
    random euclidean, $\mathbb{R}^{2}$ & $H_1$ & 500 & 0.0246 & 0.0497 & 3.76 & 7.00 & 8.20 & 1,750 \\
    grid spheres, $S^{2}$ & $H_2$ & 160 & 0.0352 & 0.0982 & 2.39 & 7.70 & 7.10 & 1,080 \\
    noisy sphere, $S^{2}$ & $H_2$ & 145 & 0.0460 & 0.106 & 4.29 & 9.00 & 13.7 & 1,930 \\
    random distance & $H_2$ & 158 & 0.352 & 58.1 & 8.59 & 37.3 & 1,050 & 3,790 \\
    random euclidean, $\mathbb{R}^{10}$ & $H_2$ & 145 & 0.0224 & 0.0480 & 0.263 & 6.30 & 5.70 & 143 \\
    \bottomrule
  \end{tabular}
  \caption{\redzedvr{} vs Ripser vs \redzed{} (no active enumeration)}
\end{table}
Interestingly, for \texttt{random\_distance}, \redzed{} still outperforms Ripser; however, it is still significantly outperformed by \redzedvr{}. In all other cases, \redzed{} is outperformed by both \redzedvr{} and Ripser by several orders of magnitude. 
This highlights where the gains come from: the improvements are due to active enumeration, but \redzed{} is a necessary step to get there. 

\section{Conclusion and future directions} \label{sec:conclusion}

In this paper, we introduce two algorithms, \redzed{} and \redzedvr{}. 
In \cref{sec:algo}, we introduce \redzed, an abstract phrasing and generalization of part of the algorithm introduced by \cite{Dey-Fan-Wang}. 
We also provide a translation between \redzed{} and the standard persistence algorithm, noting that it can be viewed as the persistence pairing algorithm with compression, along with exhaustive and retrospective reduction.

We use this abstract phrasing in \cref{sec:active}, where we discuss \textit{active enumeration}. 
Our phrasing of the algorithm makes active enumeration a natural further improvement. 
Namely, the keys in $R_n$ exactly identify the active $n$-simplices. 
The benefit of active emueration is that it identifies $(n+1)$-simplices that may contribute to $n$-th homology, allowing one to ignore almost all birth $(n+1)$-simplices, which are the main bottleneck when \redzed{} is implemented directly in this case.
In the worst case, we avoid enumerating over $97\%$ of all $(n+1)$-birth simplices, and in general this number is greater than $99\%$.
We call the algorithm with active enumeration implemented \redzedvr.
The details of our implementation of \redzedvr{} are discussed in \cref{sec:implementation}.

Finally, in \cref{sec:experiments}, we compare computation times for \redzedvr{} to Ripser \cite{bauer:ripser} in the case of Vietoris--Rips filtrations. 
In general, \redzedvr{} is both faster and uses less peak memory than Ripser. 
These results are especially pronounced in degree 1, where speedups range from 8.5 to 621 times faster, and memory improvements range from 1.1 to 71 times less memory used. 
The improvements are less pronounced, yet still present in higher degrees; however, there are cases where Ripser outperforms \redzedvr{} on either time or memory. 

There are several questions and directions that can be considered in future work.
While in \cref{sec:algo}, we developed an abstract underpinning of the persistence pairing algorithm with compression, and retrospective/exhaustive reduction, one can also try to view it through the lens of discrete Morse theory \cite{forman-morse,mischaikow-nanda-morse,scoville:discrete-morse-theory}.

Another natural direction is to identify other settings where active enumeration might apply.
These include discrete homology \cite{barcelo-capraro-white}, cubical homology \cite{kaczynski-mischaikow-mrozek:cubical}, as well as homologies of other filtrations such as Delaunay or Cech filtrations \cite{edelsbrunner-mucke-alpha-shapes, bauer-edelsbrunner-cech-delaunay}. 
Finally, it is reasonable to expect that the abstract framing of \redzed{} allows for further improvements beyond active enumeration.

\bibliographystyle{amsalphaurlmod}
\bibliography{all-refs}

%Uncomment the following if you would like an appendix
\appendix
\renewcommand{\thesection}{\Alph{section}}
%\begin{appendices}
%\section{Code for computing monoidal products}
%  \lstinputlisting[language=Python]{Graph Products.py}
%\end{appendices}
% End of appendix
%\newpage
% Uncomment the following if you have a bibliography file

\end{document}